\documentclass[journal,twoside]{IEEEtran}

\IEEEoverridecommandlockouts
\usepackage{amsmath,amssymb,amsfonts,amsthm}
\usepackage{xcolor}
\usepackage{graphicx}
\usepackage{cite}
\usepackage{color}
\usepackage{multirow}
\usepackage{mathtools}
\usepackage{hyperref}
\usepackage{soul}
\usepackage{comment}
\usepackage{chngcntr}
\usepackage{enumitem,kantlipsum}
\newtheorem{theorem}{Theorem}
\newtheorem{proposition}{Proposition}
\newtheorem{definition}{Definition}

\newtheorem{lemma}{Lemma}

\newtheorem{assumption}{Assumption}

\setulcolor{orange}

\setstcolor{red}

\begin{document}
\title{Distributed Optimization for Reactive Power\\Sharing and Stability of Inverter-Based\\Resources Under Voltage Limits}

\author{
Babak Abdolmaleki, John W. Simpson-Porco, and Gilbert Bergna-Diaz\vspace{-.1in}
\thanks{This work was supported by the Department of Electric Energy, NTNU, under Grant 81148137.}
\thanks{B. Abdolmaleki is with the Department of Electric Energy, Norwegian University of Science and Technology, 7491 Trondheim, Norway, and on leave form the Department of Electrical and Computer Engineering, University of Toronto, 10 King’s College Road, Toronto, ON, M5S 3G4, Canada (email: babak.abdolmaleki@ntnu.no).}
\thanks{J. W. Simpson-Porco is with the Department of Electrical and Computer Engineering, University of Toronto, 10 King’s College Road, Toronto, ON, M5S 3G4, Canada (email: jwsimpson@ece.utoronto.ca).}
\thanks{G. Bergna-Diaz is with the Department of Electric Energy, Norwegian University of Science and Technology, 7491 Trondheim, Norway (email: gilbert.bergna@ntnu.no).}
}
\markboth{Submitted for Publication. THIS VERSION: \today}
{Submitted for Publication. THIS VERSION: \today}

\maketitle
\begin{abstract}
Reactive power sharing and containment of voltages within limits for inverter-based resources (IBRs) are two important, yet coupled objectives in ac networks. In this article, we propose a distributed control technique to simultaneously achieve these objectives. Our controller consists of two components: a purely local nonlinear integral controller which adjusts the IBR voltage setpoint, and a distributed primal-dual optimizer that coordinates reactive power sharing between the IBRs. The controller prioritizes the voltage containment objective over reactive power sharing \textit{at all points in time}; excluding the IBRs with saturated voltages, it provides reactive power sharing among all the IBRs. Considering the voltage saturation and the coupling between voltage and angle dynamics, a formal closed-loop stability analysis based on singular perturbation theory is provided, yielding practical tuning guidance for the overall control system. To validate the effectiveness of the proposed controller for different case studies, we apply it to a low-voltage microgrid and a microgrid adapted from the CIGRE medium-voltage network benchmark, both simulated in the MATLAB/Simulink environment.
\end{abstract}
\begin{IEEEkeywords}
Distributed optimization, inverter-based resources, reactive power sharing, voltage stability.\vspace{-.1in}
\end{IEEEkeywords}
\section{Introduction}
\label{Sec:Intro}

\IEEEPARstart{P}{ower} systems are moving toward the use of more renewable energy, leading to an increasing share of inverter-based resources (IBRs) in electric networks\cite{Yunjie2022}. Together with increased power supply-demand uncertainties, this shift introduces new operation and control challenges, which in turn require new control solutions\cite{Yunjie2022,IEEEMGTF}.
Among others, proportional active and reactive power sharing among dispatchable IBRs are two important control objectives. Moreover, IBRs that are non-dispatchable in terms of active power, e.g., wind and solar units, may also participate in reactive power sharing\cite{IEEEMGTF,KhayatReview}.

Since frequency is a globally-common variable, it can be exploited to facilitate active power sharing among the IBRs\cite{Babak2019}. Voltage (magnitude), however, is not globally unique and differs from bus to bus depending on the line impedance values; therefore, it cannot be used to enforce global reactive power sharing\cite{IEEEMGTF,KhayatReview}. Reactive power flow depends most strongly on the bus voltages and their differences, in inductive networks in particular. This dependency causes an inherent trade-off between precise reactive power sharing and individual bus voltage regulation, motivating a significant volume of research work on this topic. Different centralized, decentralized, and distributed voltage and reactive power control techniques have been proposed for IBRs. The distributed techniques have attracted significant attention in power system control, especially for large-scale integration of IBRs\cite{KhayatReview}. Compared to their centralized counterparts, they rely on the exchange of information only between neighboring IBRs. In addition, they show better performance and accuracy than decentralized solutions, such as the droop control technique\cite{Simpson2017}. Therefore, it seems that the real-time distributed techniques can be a viable strategy in many situations\cite{KhayatReview,Molzahn2017}.\vspace{-.1in}

\subsection{Literature Review and Research Gaps}
\label{Subsec:LiteratureReview}
Distributed voltage and reactive power sharing control of IBRs has been studied in many papers. Some works have focused solely on the voltage regulation task and have not considered the reactive power sharing problem. The only objective in these papers is to regulate the voltages of the IBRs to a setpoint. This setpoint may be constant, or may, be updated by an external controller. For example, in\cite{Bidram2013,Shahab2020,Afshari2022,Babak2020,Chen2021,Guanglei2022} and some references therein, assuming that only a few IBRs can directly access the voltage setpoint, a leader-follower consensus algorithm is used for the IBRs to follow this setpoint which is considered a virtual leader.

Conversely, in other lines of research, reactive power sharing is considered as the main objective, and the voltage control requirements are either neglected or discussed only briefly. For example, in\cite{Schiffer2016,Fan2017,Weng2019,Li2021,Wong2020,Wong2021,Zhou2020}, distributed consensus algorithms are used to ensure a proportional reactive power sharing among the IBRs, regardless of the impacts of the controllers on the voltages. However, voltage regulation and reactive power sharing are both important, yet coupled and conflicting; therefore, they should be considered simultaneously.

Simultaneous reactive power sharing and voltage regulation has also been studied. In\cite{Bidram2014}, a leader-follower consensus-based control is proposed for reactive power sharing and voltage tracking problems, where the voltage setpoint is given by a critical bus voltage regulator. Different versions of this scheme are studied and proposed in\cite{Habibi2023,Choi2021,Raeispour2021,Ge2021}. A somewhat similar controller is proposed in\cite{Porco2015}, where unlike in \cite{Bidram2014,Choi2021,Raeispour2021,Ge2021} it is assumed that all the IBRs can directly access the voltage setpoint. These controllers, however, use a single integrator for achieving both the objectives; therefore, the accuracy of reactive power sharing and voltage regulation highly depends on the choice of control gains. The existence of the trade-off between the two objectives is discussed in\cite{Porco2015,Ge2021} as well. In\cite{Porco2015}, tuning of the control gains is suggested as a possible solution for dealing with this trade-off, while in\cite{Ge2021}, the issue is left as an open problem.

Another combination of control objectives is precise reactive power sharing and \emph{average} voltage regulation \cite{Lai2016,Lu2018,Wang2019,Nasirian2016,Shafiee2018,Zhou2018,Shi2020,Shi2021,Mohiuddin2022,Mohiuddin2020}. In this approach, instead of the individual voltages of the IBRs, their estimated average is regulated at a setpoint. To this end, in\cite{Lai2016,Lu2018,Wang2019}, the leader-follower consensus algorithm is used for average voltage regulation. Based on the leader-less consensus algorithm, a controller is proposed in\cite{Nasirian2016} where the voltage setpoint of each IBR is corrected by two terms providing average voltage regulation and accurate reactive power sharing, separately. Similarly, two other approaches are proposed in\cite{Shafiee2018,Zhou2018}, but power sharing is achieved by adjusting the droop coefficient\cite{Shafiee2018} or by changing the virtual impedance\cite{Zhou2018}. To improve the voltage profile and accuracy under input disturbances, some modified controllers are also proposed in\cite{Mohiuddin2022,Mohiuddin2020,Shi2020,Shi2021}.

While the above-mentioned control schemes can provide average voltage regulation, they may result in large deviations in the individual voltages of the IBRs, violating the limits provided by grid standards, e.g., IEEE 1547\cite{IEEEStd1547}. Therefore, in many applications, constraining the individual voltages within limits (\textit{voltage containment}) seems to be a more practical objective \cite{Mohiuddin2022,Ortmann2020,Renke2017,Mohiuddin2020}. In\cite{Mohiuddin2022,Ortmann2020}, the problem is formulated as an optimization problem, and some controllers based on the primal-dual gradient method are developed. However, these methods require knowledge of the grid model and exchange a relatively large amount of information among the IBRs. In\cite{Renke2017}, along with a consensus-based control for reactive power sharing, a leader-follower voltage containment controller is proposed to force the voltages into a safe band imposed by some minimum and maximum ``leader'' IBRs. However, the accuracy of reactive power sharing and voltage containment in this method relies on the selection of the \textit{right} leaders; i.e., one must already know which IBRs take voltages closer to the minimum and maximum limits and select them as the leader IBRs. In another attempt to bound the voltages, \cite{Mohiuddin2020} introduces a voltage variance estimation and control loop to the scheme of \cite{Nasirian2016}. However, in this method, one ``special'' IBR is left out of the reactive power sharing task so that the other units can reach simultaneous accurate reactive power sharing and bounded voltages.

Summarizing, we have observed the following research gaps. The works in\cite{Bidram2013,Shahab2020,Afshari2022,Babak2020,Chen2021,Guanglei2022,Schiffer2016,Fan2017,Weng2019,Li2021,Wong2020,Wong2021,Zhou2020} have studied \textit{either} regulation of the individual IBR voltages \textit{or} reactive power sharing, but not both, while none of the papers in\cite{Bidram2013,Shahab2020,Afshari2022,Babak2020,Chen2021,Guanglei2022,Schiffer2016,Fan2017,Weng2019,Li2021,Wong2020,Wong2021,Zhou2020,Bidram2014,Habibi2023,Choi2021,Raeispour2021,Ge2021,Porco2015,Lai2016,Lu2018,Wang2019,Nasirian2016,Shafiee2018,Zhou2018,Shi2020,Shi2021} have considered the operational IBR voltage limits. The accuracy of reactive power sharing and voltage regulation/containment under the proposed controllers in \cite{Bidram2014,Habibi2023,Choi2021,Raeispour2021,Ge2021,Porco2015,Mohiuddin2020,Renke2017,Mohiuddin2022,Ortmann2020}, depends strongly on the choice of control parameters such that if not properly designed, even when steady-state reactive power sharing under voltage limits is possible, these controllers may not provide it. Finally, a rigorous study of stability and synchronization of the power network considering the coupling between angle and voltage dynamics is absent in the above works.

\subsection{Contributions}
\label{Subsec:Contribution}
To address the observed research gap, we propose a distributed control scheme for IBRs to \textit{simultaneously} achieve voltage containment and reactive power sharing. The main contributions made in this paper are as follows.
\textbf{\textit{C1)}} In our proposed method, we make use of a \textit{distributed primal-dual optimizer} to generate a globally-unique setpoint to be tracked by a purely local nonlinear integral controller that regulates the IBR's reactive power and tunes its voltage setpoint. This architecture allows maintaining the \textit{user-defined} voltage constraints, not only in steady state but \textit{at all points in time} while ensuring that the reactive power demand is shared among the IBRs with a \textit{high accuracy}. If the above-described reactive power sharing is not possible due to saturation of the voltages, then our controller excludes only the IBRs with saturated voltages from the reactive power sharing task and allows the other IBRs, which are operating away from the voltage limits, to reach a high-accuracy reactive power sharing; i.e., the controller prioritizes voltage containment over reactive power sharing but does not punish all the IBRs. \textbf{\textit{C2)}} We analyze the system's steady state using graph theory and state its properties. Considering the coupling between voltage and angle dynamics and the voltage saturation, we rigorously study the stability of the system using the Lyapunov method (as recommended in\cite{IEEEMGTF}). To this aim, we consider a timescale separation between the dynamics of the primal-dual optimizer and the voltage-angle dynamics, conduct a singular perturbation analysis, and find the stability conditions. We also provide some practical insights into the selection of the control parameters based on the IEEE 1547 standard\cite{IEEEStd1547} and the stability analysis.
\textbf{\textit{C3)}} To validate our findings, we adapt the proposed scheme to two test systems, simulated in the MATLAB/Simulink environment. One of the systems is based on a subnetwork of the CIGRE benchmark medium-voltage distribution network.

Our first attempt to address the observed research gap was presented in\cite{BabakSEST}, where we introduced a preliminary version of our control architecture. In this paper, we extend the work in\cite{BabakSEST} in the following ways. First, we include a leakage term in the local integrator channel to provide anti-wind-up action. Second, we reformulate the selection of the integrator setpoint as an optimization problem. Third, we add a formal stability proof and a parameter selection guideline. Finally, we add a new simulation case study based on a low-voltage microgrid.

The rest of the paper is structured as follows. Section~\ref{Sec:Model} contains the system modeling and problem statement. In Section~\ref{Sec:Controller}, we introduce our proposed control scheme. We conduct steady state and stability analyses of the closed-loop system in Section~\ref{Sec:SteadyStateStabilityParameter}, where we also provide parameter selection guidelines. In Section~\ref{Sec:CaseStudy}, we present and discuss the simulation results for different case studies. Finally, Section~\ref{Sec:Conclusion} concludes the paper.\vspace{-.1in}
\section{System Modeling, Power Sharing Definition, and Droop Control Behavior}
\label{Sec:Model}
\subsection{Inverter-Based Electric Power Network}
\label{Subsec:SystemModel}
Under the hierarchical control policy\cite{IEEEMGTF,KhayatReview}, the innermost control loops of inverters are tasked with controlling the LC filter’s inductor current and capacitor voltage by generating proper switching signals (see Fig.~\ref{Fig:IBR}).
While different inner loop designs have been proposed, all are designed to act \textit{very fast}, such that the subsystem denoted by red dashed lines in Fig.~\ref{Fig:IBR} has a high bandwidth\cite{KhayatReview}; virtual impedance control can also be embedded in this subsystem to provide additional decoupling between active and reactive powers and improve system performance\cite{Wu2017VI}. For example, in our simulation case studies in Section \ref{Sec:CaseStudy}, we use the cascaded control structure described in\cite{Wu2017VI} and references therein. We also assume that the IBRs use the well-known droop control\cite{Wu2017VI} or equivalently virtual synchronous machine (VSM) control technique\cite{DArco2014} as their primary controller, which operates slowly compared with the internal control loops.
\begin{figure}
\centering
\includegraphics[width=\columnwidth]{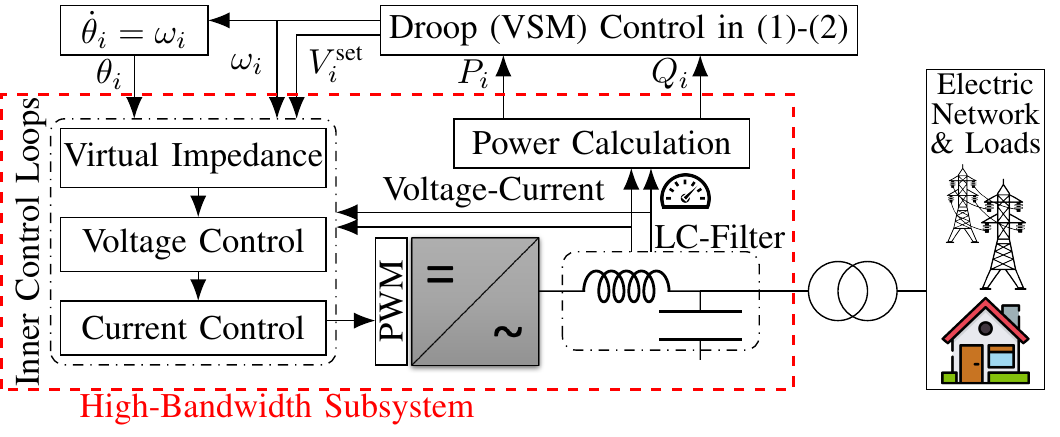}
\caption{An inverter-based resource (IBR), governed by the primary control in (1). The high-bandwidth subsystem is denoted by dashed lines. The detailed low-level control structure used in this paper can be found in, e.g.,\cite{Wu2017VI}.\label{Fig:IBR}\vspace{-.1in}}
\end{figure}

In a multi-vendor power system, however, the detailed structure and dynamics of the fast internal controllers are not easily accessible. Therefore, for high-level control design and stability studies, it is preferable to use a simplified generic model for each primary-controlled IBR\cite{WeiDu2021,Mobin2021,Babak2020TIE}. For the $i$th IBR we will use the model
\begin{IEEEeqnarray}{rCl}
\label{eq:AngleDynamics}
\dot{\theta}_i &=& \omega_i=\omega_{\rm nom}+\Omega_i\IEEEyesnumber\IEEEyessubnumber\label{eq:Freq}\\
\tau_\Omega\dot{\Omega}_i &=&-\Omega_i-  m_i^\omega P_i/S_i^{\rm rated}\IEEEyessubnumber\label{eq:FDroop}\\
\label{eq:VoltageDynamics}
V_i &=&V_i^{\rm set}= V_{\rm nom}+v_i\IEEEyesnumber\IEEEyessubnumber\label{eq:Voltage}\\
\tau_v\dot{v}_i &=&- v_i-  m_i^V Q_i/S_i^{\rm rated}\IEEEyessubnumber\label{eq:VDroop}
\end{IEEEeqnarray}
where $\theta_i$ and $\omega_i$ are the phase angle and angular frequency of the IBR, $V_i$ and $V_i^{\rm set}$ are the IBR voltage and its setpoint, and $\omega_{\rm nom}$ and $V_{\rm nom}$ are the nominal frequency and voltage. The state variables $\Omega_i$ and $v_i$ are the frequency and voltage deviations induced by droop (VSM) controllers in \eqref{eq:FDroop} and \eqref{eq:VDroop}, respectively. The constants $\tau_\Omega$ and $\tau_v$ are the frequency and voltage time constants, respectively. The constants $m_i^\omega$ and $m_i^V$ are the IBR's frequency and voltage droop coefficients, respectively. The apparent power $S_i^{\rm rated}$ is the rated capacity of the IBR while $P_i$ and $Q_i$ are respectively its active and reactive power injections, which are related to $\theta=(\theta_1,\ldots,\theta_n)$ and $V=(V_1,\ldots,V_n)$ through the following power flow equations
\begin{IEEEeqnarray}{rcl}
\label{eq:PowerFlowEquations}
P_i=f_i^P(\theta,V) &=& {\sum}_{j=1}^n V_iV_j \big( G_{ij}\cos(\theta_{ij})+B_{ij}\sin(\theta_{ij}) \big)\quad\,\,\IEEEyesnumber\IEEEyessubnumber\label{eq:PFlow}\\
Q_i=f_i^Q(\theta,V) &=& {\sum}_{j=1}^n V_iV_j \big( G_{ij}\sin(\theta_{ij})-B_{ij}\cos(\theta_{ij}) \big)\quad\,\,\IEEEyessubnumber\label{eq:QFlow}
\end{IEEEeqnarray}
where $\theta_{ij}=\theta_i-\theta_j$ is the phase difference between IBRs $i$ and $j$; $G_{ij}$ and $B_{ij}$ are the elements of the network's reduced conductance and susceptance matrices\cite[Ch. 6.4]{Kundur1994}.\vspace{-.1in}

\subsection{Power Sharing Definition and Review of Droop Control}
\label{Subsec:Problem}
In this subsection, we define power sharing among IBRs and review the steady-state behavior of droop control. As notation, for any variable $x$, let $\bar{x}$ denote its steady-state value.

\begin{definition}
\label{Def:PowerSharing}
The microgrid system \eqref{eq:AngleDynamics}--\eqref{eq:PowerFlowEquations} achieves reactive power sharing if $\bar{Q}_i/S_i^{\rm rated}=\bar{Q}_j/S_j^{\rm rated}=\alpha_Q$ for some $\alpha_Q$.
We define active power sharing similarly, using $P$ instead of $Q$.
\end{definition}

According to the droop control in \eqref{eq:FDroop} and \eqref{eq:VDroop} we have
\[
\bar{P}_i/S_i^{\rm rated}=-\bar{\Omega}_i/m_i^\omega, \qquad \bar{Q}_i/S_i^{\rm rated}=-\bar{v}_i/m_i^V.
\]
Since steady-state frequency is global, for every $i$ and $j$ we have $\bar{\Omega}_i=\bar{\Omega}_j$. Thus, following the conventional droop control design criteria\cite{Wu2017VI}, by selecting equal frequency droop coefficients for the IBRs, i.e., $m_i^\omega=m_j^\omega$, we have $\bar{P}_i/S_i^{\rm rated}=\bar{P}_j/S_j^{\rm rated}$ for every $i$ and $j$, i.e., the frequency droop controller \eqref{eq:FDroop} enforces proportional active power sharing. However, since $\bar{v}_i=\bar{v}_j$ for every $i$ and $j$ \textit{does not necessarily hold}, selecting $m_i^V=m_j^V$ does not guarantee $\bar{Q}_i/S_i^{\rm rated}=\bar{Q}_j/S_j^{\rm rated}$; i.e., the voltage droop controller \eqref{eq:VDroop} cannot enforce reactive power sharing in the same way. In what follows, we propose a distributed control scheme to provide reactive power sharing considering the IBRs voltage limits.\vspace{-.2in}
\section{Proposed Controller}
\label{Sec:Controller}
In this section, we introduce our proposed controller. The controller consists of two subsystems that will be introduced separately: \textit{a)} a nonlinear leaky integral controller for regulating reactive power ratios of the IBRs and maintaining the voltage limits, and \textit{b)} a distributed optimizer for obtaining the optimal setpoint for this integrator.\vspace{-.1in}
\subsection{Integral Reactive Power Regulation Under Voltage Limits}

Let $V_i^{\rm min}$ and $V_i^{\rm max}$ denote minimum and maximum the desired operational voltage limits for IBR $i$, with average value $V_i^\star=\tfrac{1}{2}(V_i^{\rm max}+V_i^{\rm min})$ and maximum allowable deviation $\Delta_i=\tfrac{1}{2}(V_i^{\rm max}-V_i^{\rm min})$ from that average. In place of the conventional voltage controller \eqref{eq:VoltageDynamics}, we propose the nonlinear integral controller
\begin{IEEEeqnarray}{rCl}
\label{eq:Regulator}
V_i &=&V_i^{\rm set}= V_i^\star+\Delta_i \tanh(v_i/\Delta_i),\IEEEyesnumber\IEEEyessubnumber\label{eq:ContVset}\\
\tau_v\dot{v}_i &=&V_i^\star(\lambda_i- Q_i/S_i^{\rm rated})-\beta\Delta_i \tanh(v_i/\Delta_i)-\rho_i(v_i) v_i,\,\,\,\,\quad\IEEEyessubnumber\label{eq:ContRegualtor}
\end{IEEEeqnarray}
 where $v_i$ is the state variable of the integrator \eqref{eq:ContRegualtor} with time constant $\tau_v$, and where $\beta>0$ is sufficiently small. The variable $\lambda_i$ is a setpoint for the utilization ratio $Q_i/S_i^{\rm rated}$, obtained by the optimizer, which will be subsequently described in \eqref{eq:PrimalDualDynamics}, in the next subsection. The non-negative function $\rho_i(v_i)$ is a nonlinear \textit{leakage} coefficient, defined as
 \begin{IEEEeqnarray}{rCl}
\rho_i(v_i)&=& \begin{cases}
|v_i/\Delta_i|-3&\text{if }|v_i|> 3\Delta_i\\
0 &\text{otherwise}.
\end{cases}\quad\IEEEyessubnumber\label{eq:rho}     
 \end{IEEEeqnarray}
 The main ideas behind the controller \eqref{eq:Regulator} are as follows.
\begin{itemize}[leftmargin=*]
    \item Since $\tanh$ is bounded between $-1$ and $1$, \eqref{eq:ContVset} ensures that $V_i^{\rm min}< V_i < V_i^{\rm max}$ at all points in time. In other words, voltage containment is achieved by construction.\footnote{As we will see in the stability analysis, the use of a smooth hyperbolic tangent instead of the standard saturation function, allows us to define a positive-definite Lyapunov function and facilitates the stability analysis under voltage constraints.}
    \item The first term in \eqref{eq:ContRegualtor} provides integral action for the utilization ratio $Q_i/S_i^{\rm rated}$ to track the provided setpoint $\lambda_i$. The (small) term $\beta\Delta_i \tanh(v_i/\Delta_i)$ provides damping, which will assist in our subsequent stability analysis.
    \item The nonlinear gain $\rho_i(v_i)$ in \eqref{eq:rho} prevents integrator wind-up when $|v_i| > 3\Delta_i$. The particular choice of the constant $3$ is because $|{\rm tanh}(\pm 3)| \approx 0.995$ and ${\rm tanh}(v_i/\Delta_i)$ does not change significantly for $|v_i| > 3\Delta_i$. In words, roughly speaking, for $|v_i| > 3\Delta_i$ the voltages are saturated with an acceptable accuracy.\vspace{-.1in}
\end{itemize}

\subsection{Distributed Optimization of the Integrator Setpoint $\lambda_i$}
\label{Subsec:Distributed}

In \eqref{eq:ContRegualtor}, $\lambda_i$ acts as a setpoint for the utilization ratio $Q_i/S_i^{\rm rated}$. By Definition \ref{Def:PowerSharing}, reactive power sharing will be achieved if the equilibrium values $\bar{\lambda}_{i}$ are equal, i.e., if $\bar{\lambda}_{i} = \bar{\lambda}_{j}$ for all IBRs $i$ and $j$. We now discuss the optimal selection $\bar{\lambda}_i$ for this setpoint and introduce a distributed algorithm for its online computation.

Following the above discussion, the optimal setpoint selection $\bar{\lambda}_i$ can be formulated via the following optimization problem
\begin{IEEEeqnarray}{c}
\label{eq:MultiVarProgram0}
\min_{\bar{ \lambda}_i}\,{\sum}_{i=1}^n \tfrac{1}{2}\left(\bar{\lambda}_i-\bar{ Q }_i/S_i^{\rm rated}\right)^2\IEEEyesnumber\IEEEyessubnumber\\
\text{subject to }
0 = \bar{\lambda}_i - \bar{\lambda}_j,\,\, \forall i,j.\IEEEyessubnumber\label{eq:ConsensusConstraints0}
\end{IEEEeqnarray}

We will be seeking a distributed online solution to this optimization problem. To this end, we assume that the IBRs can exchange information over a communication network modeled with an \textit{undirected (bidirectional)} and connected communication graph; see Appendix \ref{Appendix:Graph} for more info on graph theory. With $a_{ij}$ denoting the elements of the adjacency matrix and $N_i$ the set of neighbours of IBR $i$, the problem \eqref{eq:MultiVarProgram0} is equivalent to
\begin{IEEEeqnarray}{c}
\label{eq:MultiVarProgram}
\min_{\bar{\lambda}_i}\,{\sum}_{i=1}^n \tfrac{1}{2}\Big ( (\bar{\lambda}_i-\frac{\bar{ Q }_i}{S_i^{\rm rated}})^2+\tfrac{k}{2}{\sum}_{ i,j=1}^na_{ij}(\bar{\lambda}_i-\bar{\lambda}_j)^2 \Big),\quad\IEEEyesnumber\IEEEyessubnumber\label{eq:DistProgramm}\\
\text{subject to }
\bar{ z }_i= {\sum}_{j\in N_i} a_{ij}(\bar{\lambda}_i-\bar{\lambda}_j)=0,\,\, \forall i,\IEEEyessubnumber\label{eq:ConsensusConstraints}
\end{IEEEeqnarray}
where $k > 0$. The constraint \eqref{eq:ConsensusConstraints} implies that $\bar{\lambda}_i=\bar{\lambda}_j$ for all $i$ and $j$. We define the Lagrangian associated with the problem \eqref{eq:MultiVarProgram} as
\[
\mathbb{L}(\bar{ \lambda}_1,\bar{ \zeta }_1,\ldots,\bar{ \lambda}_n,\bar{ \zeta }_n) = C(\bar{ \lambda}_1,\ldots,\bar{ \lambda}_n)+ {\sum}_{i=1}^n \bar{\zeta}_i \bar{z}_i,\IEEEnonumber
\]
where $C(\bar{ \lambda}_1,\ldots,\bar{ \lambda}_n)$ is the total cost function used in \eqref{eq:DistProgramm} and $\bar{\zeta}_i$ is the Lagrange multiplier associated with the constraint $\bar{z}_i=0$. The problem \eqref{eq:MultiVarProgram} is a quadratic minimization program with linear constraints; hence, Slater's condition holds, and KKT conditions provide necessary and sufficient conditions for optimality\cite{Boyd}. In other words, $\bar{\lambda}_i$, $\bar{\zeta}_i$, and $\bar{z}_i$ are optimal if and only if they satisfy the KKT conditions \cite[Ch. 5.5]{Boyd}
\begin{IEEEeqnarray}{rCl}
\label{eq:KKTScalar}
0&=&\bar{\lambda}_i - \tfrac{\bar{ Q }_i}{S_i^{\rm rated}} + \sum_{j\in N_i} a_{ij}(\bar{ \zeta }_i-\bar{ \zeta }_j) +k \sum_{j\in N_i} a_{ij}(\bar{\lambda}_i-\bar{\lambda}_j), \qquad\IEEEyesnumber\IEEEyessubnumber\label{eq:KKTScalarStationary}\\
0&=&\bar{ z }_i =  {\sum}_{j\in N_i} a_{ij}(\bar{\lambda}_i-\bar{\lambda}_j).\IEEEyessubnumber\label{eq:KKTScalarConsesnus}
\end{IEEEeqnarray}
The solution $(\bar{\lambda}_i,\bar{\zeta}_i)$ of \eqref{eq:KKTScalar} can be computed in a distributed manner via the so-called primal-dual dynamics \cite{Cherukuri2016}
\begin{IEEEeqnarray}{rCl}
\label{eq:PrimalDualDynamics}
\tau_p\dot{\lambda}_i&=& \tfrac{Q_i}{S_i^{\rm rated}} -\lambda_i - \!\!\sum_{j\in N_i}a_{ij}(\zeta_i-\zeta_j) + k\!\!\sum_{j\in N_i}a_{ij}(\lambda_j-\lambda_i), \qquad\IEEEyesnumber\IEEEyessubnumber\label{eq:PrimalDynamics}\\
\tau_d\dot{\zeta}_i&=&  {\sum}_{j\in N_i}a_{ij}(\lambda_i-\lambda_j),\IEEEyessubnumber\label{eq:DualDynamics}
\end{IEEEeqnarray}
where $\lambda_i$ and $\zeta_i$ are now dynamic state variables which are exchanged between neighboring IBRs in real-time. The parameters $\tau_p$ and $\tau_d$ are the primal and dual dynamics time constants, which for our purposes are tunable gains. 
%
%

To summarize the overall control architecture: the subsystem \eqref{eq:PrimalDualDynamics} generates the setpoint $\lambda_i$ to be tracked by the regulator \eqref{eq:ContRegualtor}, while the regulator \eqref{eq:ContRegualtor} generates the voltage setpoint \eqref{eq:ContVset} that is saturated within limits; \eqref{eq:ContRegualtor} also provides an anti-wind-up function through the leakage term $\rho_i(v_i)v_i$ when necessary. The general scheme of the proposed controller is shown in Fig.~\ref{Fig:Controller}.\vspace{-.1in}
\begin{figure}
    \centering
    \includegraphics[width=\columnwidth]{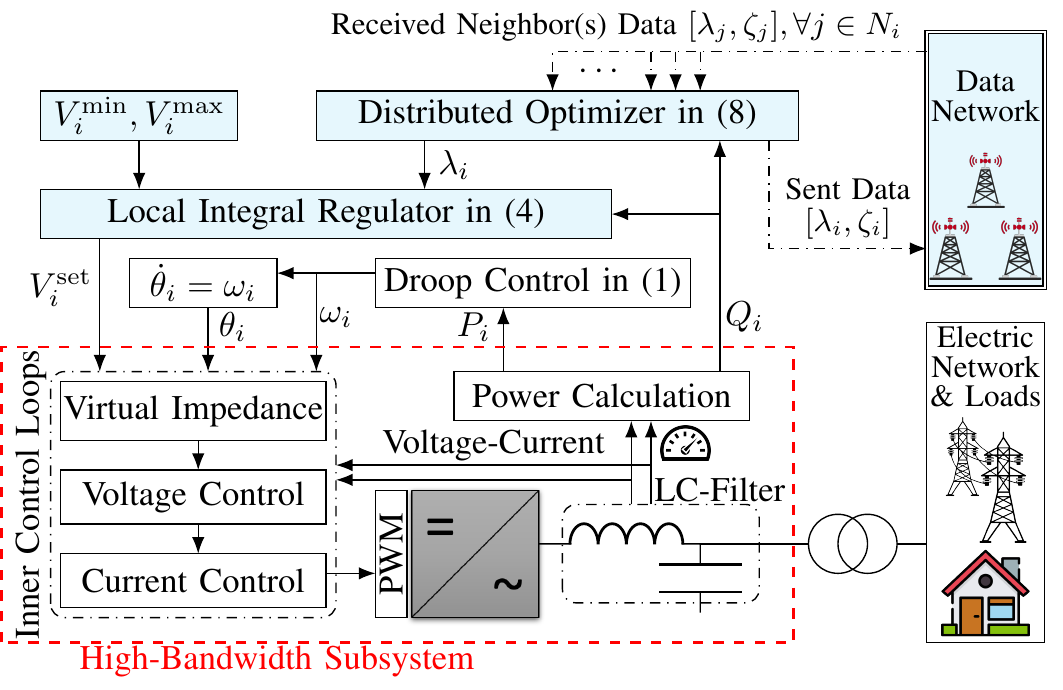}
    \caption{An IBR under the proposed controller.\label{Fig:Controller}\vspace{-.2in}}
\end{figure}


\section{Steady-State, Stability Analysis, and Controller Gain Selection}
\label{Sec:SteadyStateStabilityParameter}

The closed-loop system consists of the angle dynamics \eqref{eq:AngleDynamics}, the voltage controller \eqref{eq:Regulator} and \eqref{eq:PrimalDualDynamics}, and the power grid model \eqref{eq:PowerFlowEquations}. In this section, we analyze the steady state of the closed-loop system, study its stability, and state its properties. We also give some insights on the selection of the control parameters.\vspace{-.1in}

\subsection{System Steady State and its Properties}

We begin by writing the system dynamics in a compact form. Let $x = \mathrm{col}(x_1,\ldots,x_n)$ denote the column vector composed of elements $x_1,\ldots,x_n$. We define $m={\rm diag}(m_1^\omega,\ldots,m_n^\omega)$, $S={\rm diag}(S_1^{\rm rated},\ldots,S_n^{\rm rated})$, and $\rho (v) ={\rm diag}(\rho_1(v_1),\ldots,\rho_n(v_n))$ as well. With this, we can write the differential equations of \eqref{eq:AngleDynamics}, \eqref{eq:Regulator}, and \eqref{eq:PrimalDualDynamics} in the  compact form
\begin{IEEEeqnarray}{rCl}
\label{eq:CompactSystem}
\dot{\theta} &=& \omega=\omega_{\rm nom}1_n+\Omega,\IEEEyesnumber\IEEEyessubnumber\label{eq:AngleCompact}\\
\tau_\Omega\dot{\Omega} &=&-\Omega- m S^{-1} P,\IEEEyessubnumber\label{eq:FeqCompact}\\
\tau_v\dot{v} &=&-\rho (v) v-\beta \Delta\boldsymbol{\tanh}( \Delta^{-1}v )+[V_\star] (\lambda- S^{-1}Q ),\quad\IEEEyessubnumber\label{eq:IntCompact}\\
\tau_p\dot{\lambda}&=&-k \mathcal{L}\lambda- \mathcal{L}\zeta + (S^{-1}Q-\lambda) ,\IEEEyessubnumber\label{eq:PrimalCompact}\\
\tau_d\dot{\zeta}&=&\mathcal{L}\lambda,\IEEEyessubnumber\label{eq:DualCompact}
\end{IEEEeqnarray}
where $\mathcal{L}$ is the Laplacian matrix of the communication graph defined in Appendix~\ref{Appendix:Graph}, $1_n={\rm col}(1,\ldots,1)\in \mathbb{R}^n$, $V_\star={\rm col}(V_1^\star,\ldots,V_n^\star)$, $\Delta={\rm diag}(\Delta_1,\ldots,\Delta_n)$, $\boldsymbol{\tanh}(x)={\rm col}(\tanh(x_1),\ldots,\tanh(x_n))$, and $[V_\star]={\rm diag}(V_1^\star,\ldots,V_n^\star)$. We can also compactly write the power flow equations \eqref{eq:PowerFlowEquations} and the voltage \eqref{eq:ContVset} as
\begin{IEEEeqnarray}{rCl}
\label{eq:AlgebraicCompact}
P &=& f_P(\theta,V)={\rm col}(f_1^P(\theta,V),\ldots,f_n^P(\theta,V)),\IEEEyesnumber\IEEEyessubnumber\label{eq:PFlowCompact}\\
Q &=& f_Q(\theta,V)={\rm col}(f_1^Q(\theta,V),\ldots,f_n^Q(\theta,V)),\IEEEyessubnumber\label{eq:QFlowCompact}\\
V &=& V_\star + \Delta\boldsymbol{\tanh}( \Delta^{-1}v ).\IEEEyessubnumber\label{eq:VoltageCompact}
\end{IEEEeqnarray}

Out first result describes equilibrium points of \eqref{eq:CompactSystem}-\eqref{eq:AlgebraicCompact}.

\begin{lemma}[Steady State]
\label{Lemma:SS}

Consider the system \eqref{eq:CompactSystem}-\eqref{eq:AlgebraicCompact} and suppose that the ac power network has a synchronization frequency of $\omega_{\rm syn}$. Then any steady state of the system satisfies
\begin{IEEEeqnarray}{rCl}
\label{eq:SteadyState}
\dot{\bar{\theta}} &=& \omega_{\rm syn} 1_n=\omega_{\rm nom}1_n +\bar{\Omega},\IEEEyesnumber\IEEEyessubnumber\label{eq:SSAngle}\\
0_n &=& -\bar{\Omega}- m S^{-1} \bar{P},\IEEEyessubnumber\label{eq:SSFeq}\\
0_n &=&-\rho (\bar{v}) \bar{v}-\beta \Delta\boldsymbol{\tanh}( \Delta^{-1}\bar{v} )+[V_\star] (\bar{\lambda}- S^{-1}\bar{Q} ),\quad\IEEEyessubnumber\label{eq:SS}\\
0_n &=&  -k\mathcal{L}\bar{\lambda} -\mathcal{L}\bar{\zeta} + (S^{-1}\bar{Q}-\bar{\lambda}) ,\IEEEyessubnumber\label{eq:SSPrimal}\\
0_n &=& \mathcal{L}\bar{\lambda},\IEEEyessubnumber\label{eq:SSDual}
\end{IEEEeqnarray}
where $0_n={\rm col}(0,\ldots,0)\in \mathbb{R}^n$, $\bar{P} = f_P(\bar{\theta},\bar{V})$, $\bar{Q} = f_Q(\bar{\theta},\bar{V})$, and $\bar{V} = V_\star + \Delta\boldsymbol{\tanh}\big( \Delta^{-1}\bar{v} \big)$. Moreover, if $m_i^\omega=m_\star$ for some $m_{\star} > 0$ and all $i$, then
\begin{IEEEeqnarray}{rCll}
\label{eq:SteadyStateRelationships}
\bar{\Omega} &=& - m_\star \alpha_P 1_n, &\quad\text{ where } \alpha_P=(\tfrac{1}{n}1_n^\top S^{-1} \bar{P}) \in \mathbb{R},\quad\IEEEyesnumber\IEEEyessubnumber\label{eq:SSOmega}\\
\bar{\lambda} &=&  \alpha_Q 1_n,  &\quad\text{ where } \alpha_Q=(\tfrac{1}{n}1_n^\top S^{-1} \bar{Q}) \in \mathbb{R}.\quad\IEEEyessubnumber\label{eq:SSLambda}
\end{IEEEeqnarray}
\end{lemma}

\begin{proof}
If the ac network has a synchronization frequency of $\omega_{\rm syn}$, then we have $\dot{\bar{\theta}}=\omega_{\rm syn} 1_n$. Setting this equality together with $\dot{x}=0_n$ for any other variable $x$ in \eqref{eq:CompactSystem}, we can simply derive the steady-state equations \eqref{eq:SteadyState}. Next, we prove \eqref{eq:SteadyStateRelationships}. According to \eqref{eq:SSAngle}, we have $\bar{\Omega}=(\omega_{\rm syn}-\omega_{\rm nom})1_n$. Multiplying this equation by $1_n^\top$, we get $\omega_{\rm syn}-\omega_{\rm nom}=\tfrac{1}{n}1_n^\top\bar{\Omega}$ and hence $\bar{\Omega}=\tfrac{1}{n}1_n1_n^\top\bar{\Omega}$. On the other hand, from \eqref{eq:SSFeq} we have $\bar{\Omega}=- m_\star S^{-1} \bar{P}$, where we used $m_i^\omega=m_\star$ for all $i$. Using the last two equations, we can derive \eqref{eq:SSOmega}. By connectivity of the communication graph, every solution of equation \eqref{eq:SSDual} has the form $\bar{\lambda} = \alpha_Q 1_n$ for some $\alpha_Q\in \mathbb{R}$\cite[Ch. 6]{FB-LNS}. Since the graph is undirected, we also have $1_n^\top \mathcal{L}=0_n$\cite[Ch. 6]{FB-LNS}; multiplying \eqref{eq:SSPrimal} by $1_n^\top$ and using this property we get $1_n^\top \bar{\lambda}=1_n^\top S^{-1}\bar{Q}$. Setting $\bar{\lambda} = \alpha_Q 1_n$ in this equation, we can finally arrive at \eqref{eq:SSLambda}.
\end{proof}


Based on Lemma \ref{Lemma:SS}, we can now state several practically important properties of the steady state enforced by our controller.

\begin{proposition}[Steady State Properties]
\label{Proposition:SSProperties}

Consider a steady state as given in \eqref{eq:SteadyState}, let $\mathcal{N}$ denote the set of all the IBRs, and define the set of voltage-saturated IBRs as
\[
\mathcal{N}_{\rm sat} = \{i\in\mathcal{N} \,|\, \rho_i(\bar{v}_i)>0\}.
\]
If $m_i^\omega=m_\star$ for all $i \in \mathcal{N}$, then the steady state described by Lemma~\ref{Lemma:SS} has the following properties:
\begin{enumerate}[leftmargin=*]
    \item \textbf{Active Power Sharing and Frequency Regulation:} Active power sharing is achieved among all the IBRs and the microgrid's synchronization frequency is $\omega_{\rm syn}=\omega_{\rm nom} - m_\star \tfrac{1}{n}1_n^\top S^{-1}\bar{P}$.
    \item \textbf{Voltage Containment:} The steady-state voltages are all in the safe range, i.e., $\bar{V}_i \in (V_i^{\rm min},V_i^{\rm max})$ for all $i\in\mathcal{N}$.
    \item \textbf{Global Reactive Power Sharing:} If $\mathcal{N}_{\rm sat}=\emptyset$, then
    \[
    |\bar{Q}_i/S_{i}^{\rm rated}-\alpha_{Q}| = \beta |1-\bar{V}_i/V_i^{\star}|, \qquad \forall i\in\mathcal{N},
    \]
    i.e., the IBRs achieve reactive power sharing with a small error proportional to $\beta$.
    \item \textbf{Partial Reactive Power Sharing:} If $\mathcal{N}_{\rm sat}\neq\emptyset$, then
    \[
    |\bar{Q}_i/S_{i}^{\rm rated}-\alpha_{Q}| = \beta |1-\bar{V}_i/V_i^{\star}|, \qquad \forall i\notin\mathcal{N}_{\rm sat},\]
    \[
    |\bar{Q}_i/S_{i}^{\rm rated}-\alpha_{Q}| \leq \beta |1-\bar{V}_i/V_i^{\star}|+\rho_i (\bar{v}_i)|\bar{v}_i/V_i^\star|,\; \forall i\in\mathcal{N}_{\rm sat}.
    \]
    i.e., only the IBRs that are not in $\mathcal{N}_{\rm sat}$ achieve the described almost accurate reactive power sharing, and for the IBRs that belong to $\mathcal{N}_{\rm sat}$, the sharing accuracy decreases (deteriorates) as $\rho_i(\bar{v}_i)$ increases.
\end{enumerate}
\end{proposition}
\begin{proof}
According to \eqref{eq:SSOmega} and \eqref{eq:SSFeq}, we have $S^{-1}\bar{P}=\alpha_P 1_n$ and hence $\bar{P}_i/S_i^{\rm rated}=\bar{P}_j/S_j^{\rm rated}=\alpha_P$, for every $i$ and $j$, which according to Definition~\ref{Def:PowerSharing} underlines that active power sharing is achieved among all the IBRs. Inserting \eqref{eq:SSOmega} into \eqref{eq:SSAngle}, we can also write $\omega_{\rm syn}=\omega_{\rm nom} - m_\star \tfrac{1}{n}1_n^\top S^{-1}\bar{P}$, which proves property 1. The second property is obvious, as we have $-1 < \tanh(\cdot) < 1$. Next, we prove properties 3 and 4.

Inserting \eqref{eq:SSLambda} into \eqref{eq:SS}, and considering $\Delta\boldsymbol{\tanh}\big( \Delta^{-1}\bar{v} \big)=\bar{V} - V_\star$, we have
\begin{IEEEeqnarray}{rCl}
S^{-1}\bar{Q}&=& \alpha_Q 1_n -\beta [V_\star]^{-1}(\bar{V} - V_\star)-[V_\star]^{-1}\rho (\bar{v}) \bar{v},\qquad\IEEEyesnumber\IEEEyessubnumber\label{eq:SSURCompact}\\
\bar{Q}_i/S_i^{\rm rated} &=& \alpha_Q - \beta(\bar{V}_i/V_i^\star-1)-\rho_i (\bar{v}_i)\bar{v}_i/V_i^\star.\IEEEyessubnumber\label{eq:SSURScalar}
\end{IEEEeqnarray}
Now if $\mathcal{N}_{\rm sat}=\emptyset$, then for all $i$ we have $\rho_i (\bar{v}_i)=0$. From \eqref{eq:SSURScalar}, we can therefore write
    \begin{IEEEeqnarray}{rCl}
    \bar{Q}_i/S_i^{\rm rated} &=&\alpha_Q - \beta(\bar{V}_i/V_i^\star-1),\quad\forall i\in \mathcal{N},\quad\IEEEnonumber\label{eq:GlobalSharing}
    \end{IEEEeqnarray}
    which proves property 3. Using the definition of the set $\mathcal{N}_{\rm set}$, we can similarly write
    \begin{IEEEeqnarray}{rCl}
    \bar{Q}_i/S_i^{\rm rated} &=&\alpha_Q - \beta(\bar{V}_i/V_i^\star-1),\,\,\,\forall i\notin \mathcal{N}_{\rm sat},\qquad\IEEEnonumber\label{eq:PartialSharing}\\
    \bar{Q}_i/S_i^{\rm rated} &=& \alpha_Q - \beta(\bar{V}_i/V_i^\star-1)-\rho_i (\bar{v}_i)\bar{v}_i/V_i^\star,\,\,\,\forall i\in \mathcal{N}_{\rm sat},\quad\IEEEnonumber\label{eq:PoorSharing}
    \end{IEEEeqnarray}
    which, according to Definition~\ref{Def:PowerSharing}, proves property 4.\vspace{-.1in}
\end{proof}
\subsection{Stability Analysis}
We want to analyze the stability for the system \eqref{eq:CompactSystem}. To this end, we first take some steps to simplify the system dynamics and then analyze stability for the simplified version of \eqref{eq:CompactSystem}.

As our controller will maintain voltages within limits around their nominal values, it seems reasonable to use a linearized power flow model to describe the network behavior around the operating point.

\begin{assumption}
\label{Assumption:LinearPowerFlow}
Around a nominal operating point, the power flow equations in \eqref{eq:AlgebraicCompact} can be approximated by
\begin{IEEEeqnarray}{rCl}
\label{eq:LinearPowerFlow}
P &=& J_\theta^P \theta + J_V^P V + w_P,\IEEEyesnumber\IEEEyessubnumber\label{eq:PLinear}\\
Q &=& J_\theta^Q \theta + J_V^Q V + w_Q,\IEEEyessubnumber\label{eq:QLinear}
\end{IEEEeqnarray}
where $J_\theta^P$ and $J_V^P$ (resp. $J_\theta^Q$ and $J_V^Q$) are the $n\times n$ Jacobian matrices of $f_P(\theta,V)$ (resp. $f_Q(\theta,V)$) with respect to $\theta$ and $V$ at the linearization point, respectively; $w_P$ and $w_Q$ are the corresponding intercepts of the linear functions. The matrices $J_{\theta}^{P}$ and $J_{\theta}^{Q}$ each have an eigenvalue at $0$ with corresponding right eigenvector $1_{n}$.
\end{assumption}

We next reduce the order of the system and transform it into relative coordinates, which allows us to leverage \textit{singular perturbation analysis}\cite[Ch. 11]{Khalil} and find the stability conditions.

\subsubsection{\textbf{Model Reduction and Coordinate Transformation}}
As they are tunable control parameters, we can make the following assumption about the time constants $\tau_\Omega$, $\tau_p$, $\tau_d$, $\tau_v$ in \eqref{eq:CompactSystem}.
\begin{assumption}
\label{Assumption:FastLowPassFilter}
We have $\tau_\Omega,\tau_p<\!\!<\tau_v$ and $\tau_p<\!\!<\tau_d$.
\end{assumption}
According to the low-pass filters \eqref{eq:FeqCompact} and \eqref{eq:PrimalCompact}, we have $\Omega =- m S^{-1} P - \tau_\Omega\dot{\Omega}$ and $(I_n + k \mathcal{L})\lambda=- \mathcal{L}\zeta + S^{-1}Q - \tau_p\dot{\lambda}$. Under Assumption~\ref{Assumption:FastLowPassFilter}, the terms $\tau_\Omega\dot{\Omega}$ and $\tau_p\dot{\lambda}$ can be viewed as some negligible parasitic effects; therefore, the system dynamics are mainly governed by \eqref{eq:AngleCompact}, \eqref{eq:IntCompact}, and \eqref{eq:DualCompact}. Indeed, one may apply singular perturbation theory to rigorously reduce the order of the system dynamics to the dynamics of $\theta$, $v$, and $\zeta$ (for an example, see \cite{DorflerOverDamped}); instead, we omit the details and simply eliminate the left-hand sides of equations \eqref{eq:FeqCompact} and \eqref{eq:PrimalCompact} and consider $\Omega =- m S^{-1} P$ and $(I_n + k \mathcal{L})\lambda=- \mathcal{L}\zeta + S^{-1}Q$. Therefore, considering the linearized power flow equations in Assumption~\ref{Assumption:LinearPowerFlow}, the system \eqref{eq:CompactSystem} reduces to
\begin{IEEEeqnarray}{rCl}
\label{eq:SystemAfterAssumptions}
\dot{\theta} &=& \omega_{\rm nom}1_n + \Omega,\IEEEyesnumber\IEEEyessubnumber\label{eq:MainAngleCompact}\\
\tau_v\dot{v} &=&-\beta V + [V_\star] (\lambda- S^{-1}(J_\theta^Q \theta + J_V^Q V + w_Q) )\IEEEnonumber\\
&&-\rho (v) v + \beta V_\star,\IEEEyessubnumber\label{eq:MainIntCompact}\\
\varepsilon\dot{\zeta}&=&-\tau_v^{-1}\mathcal{L} \mathcal{K}\mathcal{L}\zeta +\tau_v^{-1}\mathcal{L} \mathcal{K}S^{-1}(J_\theta^Q \theta + J_V^Q V + w_Q),\qquad\IEEEyessubnumber\label{eq:MainDualCompact}\\
\Omega&=&- m S^{-1} (J_\theta^P \theta + J_V^P V + w_P),\IEEEyessubnumber\label{eq:OmegaCompact}\\
\lambda&=&- \mathcal{K}\mathcal{L}\zeta + \mathcal{K}S^{-1}(J_\theta^Q \theta + J_V^Q V + w_Q) ,\IEEEyessubnumber\label{eq:LambdaCompact}
\end{IEEEeqnarray}
where $\mathcal{K}=(I_n+k\mathcal{L})^{-1}$ and $\varepsilon=\tau_d/\tau_v$.

We now want to study the stability of the steady state for the system \eqref{eq:SystemAfterAssumptions} via \textit{singular perturbation analysis}. In particular, the analysis in \cite[Theorem 11.3]{Khalil} requires a system evolving on Euclidean space and an exponentially stable fixed point for the fast dynamics. In order to satisfy these requirements, we instead analyze a version of the system \eqref{eq:SystemAfterAssumptions} in which the system is transformed into relative coordinates.
Let us now define the change of coordinates $x_\theta=T\theta=[\theta_{\rm av}\:r_\theta^\top]^\top$ and $x_\zeta=T\zeta=[\zeta_{\rm av} \: r_\zeta^\top]^\top$, where $\theta_{\rm av}$ and $\zeta_{\rm av}$ are respectively the average values of the elements in $\theta$ and $\zeta$, the vectors $r_\theta$ and $r_\zeta$ belong to $\mathbb{R}^{n-1}$, and the transformation matrix $T\in \mathbb{R}^{n\times n}$ is
\begin{IEEEeqnarray}{c}
\label{eq:TMatrixProperties}
T = \begin{bmatrix}
1/n & 1/n & \ldots &1/n\\
-1 & 1 & & \\
& \ddots & \ddots & \\
 & & -1 & 1
\end{bmatrix},\quad T1_n=\begin{bmatrix}
1\\
0\\
\vdots\\
0
\end{bmatrix}\in\mathbb{R}^n.\IEEEyesnumber\IEEEyessubnumber\label{eq:TMatrix}\qquad
\end{IEEEeqnarray}
By Assumption \ref{Assumption:LinearPowerFlow} and connectivity of the communication graph, the matrices $J_\theta^P$, $J_\theta^Q$, and $\mathcal{L}$ satisfy $J_\theta^P 1_n=J_\theta^Q 1_n=\mathcal{L} 1_n=0_n$ \cite{SimpsonPorco2016}\cite[Ch. 6]{FB-LNS}.
Using $T1_n$ in \eqref{eq:TMatrix} and these properties, we compute that
\begin{IEEEeqnarray}{rCl}
T J_\theta^P T^{-1}&=&T_\theta^P=\begin{bmatrix}
0 & c_P^\top\\
0_{n-1} & J_{\theta{\rm red}}^P
\end{bmatrix},\IEEEyessubnumber\label{eq:TMatrixP}\\
T J_\theta^Q T^{-1}&=&T_\theta^Q=\begin{bmatrix}
0 & c_Q^\top\\
0_{n-1} & J_{\theta{\rm red}}^Q
\end{bmatrix},\IEEEyessubnumber\label{eq:TMatrixQ}\\
T \mathcal{L} T^{-1}&=&T_\zeta=\begin{bmatrix}
0 & c_\zeta^\top\\
0_{n-1} & \mathcal{L}_{\rm red}
\end{bmatrix},\IEEEyessubnumber\label{eq:TMatrixZeta}
\end{IEEEeqnarray}
for $J_{\theta{\rm red}}^P,J_{\theta{\rm red}}^Q,\mathcal{L}_{\rm red} \in \mathbb{R}^{(n-1)\times(n-1)}$ and $c_P,c_Q,c_\zeta \in\mathbb{R}^{n-1}$.
Let us now use $x_\theta=T\theta,\,x_\zeta=T\zeta$ and write the system dynamics \eqref{eq:SystemAfterAssumptions} in the new coordinates as
\begin{IEEEeqnarray}{rCl}
\label{eq:NewCoordinates}
\dot{x}_\theta &=&- Tm S^{-1}  T^{-1}T_\theta^P x_\theta - Tm S^{-1} J_V^P V \IEEEnonumber\\
&& + \omega_{\rm nom}T1_n  - Tm S^{-1} w_P,\qquad\IEEEyesnumber\IEEEyessubnumber\label{eq:NewAngle}\\
\tau_v\dot{v} &=&[V_\star](\mathcal{K}-I_n)S^{-1} T^{-1} T_\theta^Q x_\theta - [V_\star]\mathcal{K}T^{-1} T_\zeta x_\zeta  \IEEEnonumber\\
&& + [V_\star](\mathcal{K}-I_n)S^{-1} J_V^Q V -\beta V-\rho (v) v \IEEEnonumber\\
&&+ \beta V_\star + [V_\star](\mathcal{K}-I_n)S^{-1} w_Q,\IEEEyessubnumber\label{eq:NewVoltage}\\
\varepsilon\dot{x}_\zeta&=&\tau_v^{-1}T\mathcal{L} \mathcal{K}S^{-1}  T^{-1} T_\theta^Q x_\theta +\tau_v^{-1}T\mathcal{L} \mathcal{K}S^{-1} J_V^Q V \IEEEnonumber\\
&&-\tau_v^{-1}T\mathcal{L} \mathcal{K}T^{-1} T_\zeta x_\zeta  +\tau_v^{-1}T\mathcal{L} \mathcal{K}S^{-1} w_Q.\qquad\IEEEyessubnumber\label{eq:NewDual}
\end{IEEEeqnarray}
According to \eqref{eq:TMatrixProperties}, the first columns of $T_\theta^P$, $T_\theta^Q$, and $T_\zeta$ are all zeros, which means the first elements of $x_\theta$ and $x_\zeta$ -- $\theta_{\rm av}$ and $\zeta_{\rm av}$ -- do not influence the dynamics in \eqref{eq:NewCoordinates} at all. Therefore, \eqref{eq:NewCoordinates} is the interconnection of the two cascaded subsystems, given by
\begin{IEEEeqnarray}{rCl}
\label{eq:GroundedSystem}
\dot{r}_\theta &=&  R_\theta r_\theta +R_{\theta V} V + d_\theta,\IEEEyesnumber\IEEEyessubnumber\label{eq:GroundedAngle}\\
\tau_v\dot{v} &=&R_{v\theta}r_\theta + (R_{vV}-\beta I_n) V +R_{v\zeta} r_\zeta -\rho (v) v + d_v  ,\qquad\IEEEyessubnumber\label{eq:GroundedInt}\\
\varepsilon\dot{r}_\zeta&=&R_{\zeta \theta}r_\theta +R_{\zeta V} V + R_\zeta r_\zeta + d_\zeta,\IEEEyessubnumber\label{eq:GroundedDual}\\
\label{eq:AverageDynamics}
\dot{\theta}_{\rm av} &=& R_\theta^{\rm av} r_\theta +R_{\theta V}^{\rm av} V + d_\theta^{\rm av}, \IEEEyesnumber\IEEEyessubnumber\label{eq:AverageDynamicAngle}\\
\varepsilon\dot{\zeta}_{\rm av} &=& 0_n,\IEEEyessubnumber\label{eq:AverageDynamicZeta}
\end{IEEEeqnarray}
where their components are given in Appendix~\ref{Appendix:Components}. It should be noted that to obtain \eqref{eq:GroundedSystem}-\eqref{eq:AverageDynamics} from the dynamics \eqref{eq:NewCoordinates}, we have used the properties $r_\theta=I_r x_\theta$, $r_\zeta=I_r x_\zeta$, $\theta_{\rm av}=1_n^\top T^\top x_\theta$, $\zeta_{\rm av}=1_n^\top T^\top x_\zeta$, $x_\theta=I_r^\top r_\theta + T1_n\theta_{\rm av}$, and $x_\zeta=I_r^\top r_\zeta + T1_n\zeta_{\rm av}$, where $I_{r}=[0_{n-1}\, I_{n-1}]\in \mathbb{R}^{(n-1)\times n}$.

Clearly, the dynamics of $r_\theta$, $r_\zeta$, and $v$ do not depend on $\theta_{\rm av}$ and $\zeta_{\rm av}$. Therefore, the steady states of \eqref{eq:NewCoordinates} and hence \eqref{eq:SystemAfterAssumptions} are stable, \textit{if and only if} the steady state of \eqref{eq:GroundedSystem} is stable. In what follows, we discover this.

\subsubsection{\textbf{Timescale Separation and Singular Perturbation Analysis}}
We are now interested in studying the stability of the steady state of the system \eqref{eq:GroundedSystem} using the idea of timescale separation by considering \eqref{eq:GroundedAngle}-\eqref{eq:GroundedInt} as the slow dynamics and \eqref{eq:GroundedDual} as the fast dynamics. The following theorem states the stability conditions under these considerations. 
\begin{theorem}[Exponential Stability for \eqref{eq:GroundedSystem}]
\label{Theorem:Stability}
Suppose that the linear matrix inequality
\begin{IEEEeqnarray}{c}
\mathcal{P}_\theta \succ 0, \quad \mathcal{D}_v \succ 0,\quad \mathcal{Q}+\mathcal{Q}^\top \prec 0,\IEEEyesnumber\IEEEyessubnumber\label{eq:LMI}
\end{IEEEeqnarray}
in the variables $\mathcal{P}_\theta$ and $\mathcal{D}_v$ has a solution, where $\mathcal{P}_\theta$ is symmetric, $\mathcal{D}_v$ is diagonal, and $\mathcal{Q}$ is
\begin{IEEEeqnarray}{rCl}
\mathcal{Q} &=& \begin{bmatrix}
 \mathcal{P}_\theta R_\theta & \mathcal{P}_\theta R_{\theta V} \\
\mathcal{D}_v R_{v\theta}^{\rm new}  & \mathcal{D}_v (R_{vV}^{\rm new}-\beta I_n)
\end{bmatrix},\qquad\IEEEyessubnumber\label{eq:MatrixQ}\\
\text{where }&&\begin{cases}
R_{v\theta}^{\rm new}=R_{v\theta}-R_{v\zeta} R_\zeta^{-1}R_{\zeta \theta}\\
R_{vV}^{\rm new}=R_{vV}-R_{v\zeta} R_\zeta^{-1}R_{\zeta V}.
\end{cases}\IEEEyessubnumber\label{eq:MatrixQComponents}
\end{IEEEeqnarray}
Then, there exists $\varepsilon^\star>0$ such that for all $\tau_d<\varepsilon^\star\tau_v$ the steady state of the system \eqref{eq:GroundedSystem} is exponentially stable.
\end{theorem}
\begin{proof}
We consider $\varepsilon=\tau_d/\tau_v$ small and \eqref{eq:GroundedDual} as the fast dynamics; therefore, the velocity $\dot{r}_\zeta \propto (1/\varepsilon)$ can be large when $\varepsilon$ is small and ${r}_\zeta$ in \eqref{eq:GroundedDual} may rapidly converge to a root of $R_{\zeta \theta}r_\theta +R_{\zeta V} V + R_\zeta r_\zeta + d_\zeta=0_n$. In other words, the subsystem \eqref{eq:GroundedDual} may quickly achieve a \textit{quasi-steady state}, where $r_\zeta \approx -R_\zeta^{-1}(R_{\zeta \theta}r_\theta +R_{\zeta V} V + d_\zeta)$. We now define the error between the actual $r_\zeta$ and this quasi-steady state as $y=r_\zeta+R_\zeta^{-1}(R_{\zeta \theta}r_\theta +R_{\zeta V} V + d_\zeta)$. We can therefore write \eqref{eq:GroundedSystem} as the \textit{singular perturbation problem} below\cite[Ch. 11]{Khalil}.
\begin{IEEEeqnarray}{rCl}
\label{eq:SPP}
\dot{r}_\theta &=&  R_\theta r_\theta +R_{\theta V} V + d_\theta,\IEEEyesnumber\IEEEyessubnumber\label{eq:SPPAngle}\\
\tau_v\dot{v} &=&R_{v\theta}^{\rm new}r_\theta + (R_{vV}^{\rm new}-\beta I_n) V  +R_{v\zeta} y\IEEEnonumber\\
&&-\rho (v) v + d_v^{\rm new}  ,\IEEEyessubnumber\label{eq:SPPInt}\\
\varepsilon\dot{y}&=&R_\zeta y +\varepsilon R_\zeta^{-1}(R_{\zeta \theta}\dot{r}_\theta +R_{\zeta V} (\partial V/\partial v)\dot{v}),\IEEEyessubnumber\label{eq:SPPDual}\\
V &=& V_\star + \Delta\boldsymbol{\tanh}( \Delta^{-1}v ),\IEEEyessubnumber\label{eq:GroundedVoltage}
\end{IEEEeqnarray}
where $R_{v\theta}^{\rm new}$ and $R_{vV}^{\rm new}$ are given in \eqref{eq:MatrixQComponents} and $d_v^{\rm new}=d_v-R_{v\zeta} R_\zeta^{-1} d_\zeta$. We want to examine the stability of the steady state of \eqref{eq:SPP} by examining the reduced system \eqref{eq:ReducedAngle}-\eqref{eq:ReducedInt} and boundary-layer system \eqref{eq:BoundaryLayer}, given as
\begin{IEEEeqnarray}{rCl}
\label{eq:ReducedBLDynamics}
\dot{r}_\theta &=&  R_\theta r_\theta +R_{\theta V} V + d_\theta,\IEEEyesnumber\IEEEyessubnumber\label{eq:ReducedAngle}\\
\tau_v\dot{v} &=&R_{v\theta}^{\rm new}r_\theta + (R_{vV}^{\rm new}-\beta I_n) V  -\rho (v) v + d_v^{\rm new}  ,\quad\IEEEyessubnumber\label{eq:ReducedInt}\\
\partial y/\partial \mathbf{t}&=&R_\zeta y,\IEEEyessubnumber\label{eq:BoundaryLayer}
\end{IEEEeqnarray}
where $\mathbf{t}=t/\varepsilon$ is a stretched timescale with $t$ the time. From Appendix~\ref{Appendix:Components}, we have $R_\zeta=-  I_r \tau_v^{-1}T\mathcal{L} \mathcal{K} \mathcal{L} T^{-1} I_r^\top$. By the connectivity of the communication graph and the definitions of $\mathcal{K}$, $I_r$, and $T$, one can establish that the matrix $R_\zeta$ is negative-definite; hence, there exists a symmetric matrix $\mathcal{P}_y\succ 0$ such that $\mathcal{P}_y R_\zeta+R_\zeta^\top \mathcal{P}_y \prec 0$. With the matrix $\mathcal{P}_y$ and the matrices $\mathcal{P}_\theta \succ 0$ and $\mathcal{D}_v \succ 0$ in \eqref{eq:LMI}, we now take the following Lyapunov candidates for the slow \eqref{eq:ReducedAngle}-\eqref{eq:ReducedInt} and fast \eqref{eq:BoundaryLayer} dynamics.
\begin{IEEEeqnarray}{rCl}
\mathcal{S}_s(r_\theta,v)
&=& \tfrac{1}{2} \tilde{r}_\theta^\top\mathcal{P}_\theta\tilde{r}_\theta
+ \tau_v\int_{0_n}^{\tilde{v}}  ( \tilde{h} (\tau) )^\top \mathcal{D}_v d\tau,\IEEEyesnumber\IEEEyessubnumber\label{eq:LyapunovReduced}\\
\mathcal{S}_f(y)&=& \tfrac{1}{2} \tilde{y}^\top  \mathcal{P}_y \tilde{y},\IEEEyessubnumber\label{eq:LyapunovBL}
\end{IEEEeqnarray}
where $\tilde{r}_\theta=r_\theta-\bar{r}_\theta$, $\tilde{v}=v-\bar{v}$, and $\tilde{y}=y-\bar{y}$ with $\bar{r}_\theta$, $\bar{v}$, and $\bar{y}$ the steady states in \eqref{eq:ReducedBLDynamics}, and $\tilde{h} (\tau)$ the following function 
\begin{IEEEeqnarray}{rCl}
\tilde{h} (\tau)&=&\Delta\boldsymbol{\tanh}\big( \Delta^{-1}(\bar{v}+\tau) \big)-\Delta\boldsymbol{\tanh}\big( \Delta^{-1}(\bar{v}) \big),\IEEEyesnumber\quad\label{eq:ShiftedHyperbolic}
\end{IEEEeqnarray}
which is \textit{element-wise strictly increasing} in $\tau$. Time derivatives of $\mathcal{S}_s$ and $\mathcal{S}_f$ are
\begin{IEEEeqnarray}{rCl}
\label{eq:DrivativeLyapunov}
\dot{\mathcal{S}}_s&=&
\tilde{r}_\theta^\top\mathcal{P}_\theta\dot{\tilde{r}}_\theta
+( \tilde{h} (\tilde{v}) )^\top \mathcal{D}_v  \tau_v\dot{\tilde{v}},\IEEEyesnumber\IEEEyessubnumber\label{eq:DerivativeLyapunovReduced}\\
\partial\mathcal{S}_f/\partial \mathbf{t}&=& \tilde{y}^\top \mathcal{P}_y    \partial \tilde{y}/\partial \mathbf{t}.\IEEEyessubnumber\label{eq:DerivativeLyapunovBL}
\end{IEEEeqnarray}
Inserting the dynamics \eqref{eq:ReducedBLDynamics} into \eqref{eq:DrivativeLyapunov}, we have
\begin{IEEEeqnarray}{rCl}
\label{eq:SupplyRates}
\dot{\mathcal{S}}_s&=& -( \tilde{h} (\tilde{v}) )^\top \mathcal{D}_v \tilde{\delta}(\tilde{v})
+\tilde{\eta}^\top \mathcal{Q}\tilde{\mathcal{\eta}},\quad\IEEEyesnumber\IEEEyessubnumber\label{eq:SupplyRateReduced}\\
\partial\mathcal{S}_f/\partial \mathbf{t}&=&   \tilde{y}^\top \mathcal{P}_y R_\zeta \tilde{y},\IEEEyessubnumber\label{eq:SupplyRateBL}
\end{IEEEeqnarray}
where $\tilde{\delta}(\tilde{v})=\rho(v)(v)-\rho(\bar{v})\bar{v}$, $\tilde{\eta}={\rm col}(\tilde{r}_\theta,\tilde{h} (\tilde{v}))$ and $\mathcal{Q}$ is as given in \eqref{eq:MatrixQ}. The functions $\tilde{\delta}(\tilde{v})$ and $\tilde{h}(\tilde{v})$ are both element-wise increasing with respect to $\tilde{v}$; therefore, we have
\begin{IEEEeqnarray}{rCl}
\label{eq:Inequalities}
-( \tilde{h} (\tilde{v}) )^\top  \mathcal{D}_v \tilde{\delta}(\tilde{v}) &\leq& 0. \IEEEyesnumber\IEEEyessubnumber\label{eq:Monotonicity}
\end{IEEEeqnarray}
If the matrix inequality \eqref{eq:LMI} holds, we can write
\begin{IEEEeqnarray}{CrCl}
\tilde{\eta}^\top \mathcal{Q} \tilde{\mathcal{\eta}} \leq -\alpha_s \tilde{\eta}^\top\tilde{\eta},\IEEEyessubnumber\label{eq:LMIResult}
\end{IEEEeqnarray}
where $\alpha_s>0$ is the smallest eigenvalue of $-(\mathcal{Q}+\mathcal{Q}^\top)$. We can also write
\begin{IEEEeqnarray}{CrCl}
\tilde{y}^\top \mathcal{P}_y R_\zeta \tilde{y} = \tilde{y}^\top (\mathcal{P}_y R_\zeta+R_\zeta^\top \mathcal{P}_y) \tilde{y} < -\alpha_f \tilde{y}^\top \tilde{y} ,\IEEEyessubnumber\label{eq:ConnectivityResult}
\end{IEEEeqnarray}
where $\alpha_f>0$ is the smallest eigenvalue of $-(\mathcal{P}_y R_\zeta+R_\zeta^\top \mathcal{P}_y)$. Using \eqref{eq:Monotonicity}-\eqref{eq:ConnectivityResult}, we can now bound the derivatives in \eqref{eq:SupplyRates} as
\begin{IEEEeqnarray}{c}
\dot{\mathcal{S}}_s \leq -\alpha_s \tilde{\eta}^\top\tilde{\eta},\qquad \partial\mathcal{S}_f/\partial \mathbf{t} \leq -\alpha_f \tilde{y}^\top \tilde{y},\IEEEyesnumber\label{eq:SupplyRatesBound}
\end{IEEEeqnarray}
which, together with the fact that $\mathcal{S}_s$ and $\mathcal{S}_f$ are both positive-definite and radially unbounded, show exponential stability of the steady states of the reduced and boundary-layer dynamics. Now, we can say that the singularly perturbed system \eqref{eq:SPP} satisfies all the assumptions of \cite[Theorem 11.3]{Khalil}; therefore, there exists $\varepsilon^\star>0$ such that for all $\varepsilon<\varepsilon^\star$ or equivalently $\tau_d<\varepsilon^\star\tau_v$, the steady state of \eqref{eq:SPP} is exponentially stable.\vspace{-.1in}
\end{proof}
\subsection{Intuition on Parameter Selection}

The tunable control parameters in \eqref{eq:AngleDynamics}, \eqref{eq:VoltageDynamics}, \eqref{eq:Regulator}, and \eqref{eq:PrimalDualDynamics} are $m_i^\omega$, $m_i^V$, $\tau_\Omega$, $\tau_p$, $\tau_v$, $\tau_d$, $k$, and $\beta$. Following the standard droop control design\cite{Wu2017VI}, we select the droop coefficients $m_i^V=\Delta_i$ and $m_i^\omega=m_\star=2\pi \Delta f_{\rm max}$ for all $i$, where $\Delta f_{\rm max}$ is the maximum steady-state frequency deviation. In what follows, we introduce the impacts and limitations of the remaining parameters and propose a selection procedure.
\begin{enumerate}[leftmargin=*]
    \item $(\tau_p, \tau_\Omega)$: According to Assumption~\ref{Assumption:FastLowPassFilter}, our stability analysis is based on $\tau_\Omega,\tau_p<\!\!<\tau_v$ and $\tau_p<\!\!<\tau_d$; therefore, the smaller $\tau_\Omega$ and $\tau_p$, the more reliable our stability analysis. We suggest starting the selection procedure by selecting a small time constant for the low-pass filter \eqref{eq:PrimalDynamics}, e.g., $\tau_p=0.01s$. One can, however, select a larger $\tau_p$ for better filtering, if required. On the other hand, according to \eqref{eq:Freq}, decreasing the frequency time constant $\tau_\Omega$ increases $|\dot{f}_i|=\tfrac{1}{2\pi}|\dot{\omega}_i|=\tfrac{1}{2\pi}|\dot{\Omega}_i|$, known as the \textit{Rate of Change of Frequency} (\texttt{RoCoF}), which should be limited in practice (see, for example, \cite[Table 21]{IEEEStd1547}). Thus, we suggest selecting $\tau_\Omega = m_i^\omega/(2\pi \texttt{RoCoF}^\star)$, where $\texttt{RoCoF}^\star$ is the maximum withstandable initial \texttt{RoCoF} after a step change in $P_i$ from $0$ to $S_i^{\rm rated}$ or vice versa.
    
    \item $(\tau_v,\tau_d)$: Following the discussion in the previous step, to make our stability analysis more reliable we suggest selecting $\tau_d>\!\!>\tau_p$ and $\tau_v>\!\!>\tau_\Omega,\tau_p$, for example, $\tau_d \geq 10\tau_p$ and $\tau_v \geq \max\{10\tau_\Omega,10\tau_p\}$. On the other hand, from Theorem~\ref{Theorem:Stability}, one should select $\tau_d<\varepsilon^\star \tau_v$. Since the exact value of $\varepsilon^\star$ is not easily available, we select $\tau_v$ as large as possible and $\tau_d$ as small as possible, for example, we suggest selecting $\tau_v\geq 10 \tau_d$. Combining these suggestions, we get $\tau_d \geq 10\tau_p$ and $\tau_v \geq \max\{10\tau_\Omega,10\tau_d\}$. Here, it should be noted that, in practice, a small $\tau_d$ requires fast (low-latency) inter-IBR data transmissions. But, selecting a large $\tau_d$ leads to a large $\tau_v$, which in turn causes slower regulation of the IBR voltage and its reactive power. Therefore, while selecting the control parameters, we should consider the practical standards, e.g., IEEE 1547 \cite{IEEEStd1547}, on this matter. For example, according to \cite[Ch. 5.3]{IEEEStd1547} the response time for voltage-reactive power control, depending on the mode and application, varies between $1$ to $10$ seconds. Selecting $\tau_\Omega,\tau_d \in [0.1s,1s]$ and following the above suggestions, we get $\tau_v\in[1s,10s]$ which lies in this acceptable range.

    \item $(k,\beta)$: Clearly, $\beta$ helps solvability of the linear matrix inequality \eqref{eq:LMI}; it increases the eigenvalues of $-(\mathcal{Q}+\mathcal{Q}^\top)$ and hence the convergence rate $\alpha_s$ in \eqref{eq:SupplyRatesBound}. But, according to Proposition~\ref{Proposition:SSProperties}, it degrades the steady-state reactive power sharing. Therefore, we suggest selecting a desired $k>0$ first\footnote{We suggest selecting a desired $k_d>0$ and computing $k = k_d/\sigma_2$, where $\sigma_2$ is the second smallest eigenvalue of $\mathcal{L}$, known as algebraic connectivity of the communication graph\cite[Ch. 6]{FB-LNS}.} and then selecting a small $\beta$ such that: \textbf{\textit{i)}} the linear matrix inequality \eqref{eq:LMI} has a solution, and \textbf{\textit{ii)}} for every IBR $i$, the value $\beta \Delta_i /V_i^\star$ is an acceptable upper bound of the error $|\bar{Q}_i/S_i^{\rm rated}-\alpha_Q|$.
\end{enumerate}\vspace{-.1in}
\section{Case Studies and Simulation Results}
\label{Sec:CaseStudy}
\begin{figure}
\centering
\includegraphics[width=.85\columnwidth]{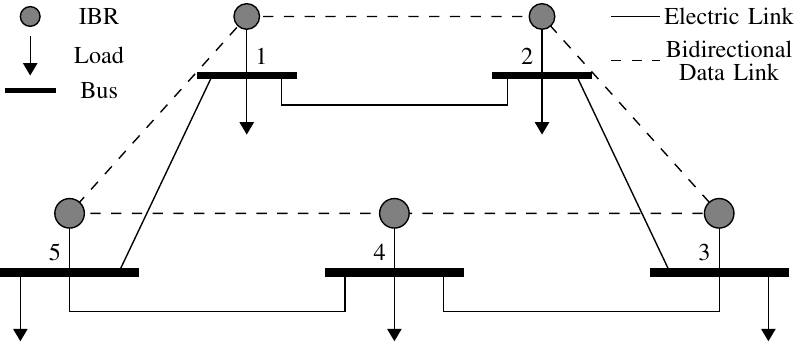}
\caption{Test low-voltage microgrid system with specifications given in Table~\ref{Table:LVMG}.\label{Fig:LVMG}\vspace{-.1in}}
\end{figure}
\begin{table}
\centering
\caption{Control and Electric Specifications for the LV system in Fig.~\ref{Fig:LVMG}}
\label{Table:LVMG}
\begin{tabular}{|c|c|c|c|}\hline
$V_{\rm nom}$ & $(V_i^{\rm min},V_i^{\rm max})$ & $S_{\rm Base}$& $(\tau_\Omega,m_i^\omega,m_i^V)$\\ \hline
220 [V] & (0.95, 1.05) [p.u.] & 100 [kVA]& (0.1, 1.57, 11)\\ \hline
$f_{\rm nom}$ & $\Delta f_{\rm max}$ & $(\tau_v,\tau_d,\tau_p)$& $(\beta,k_d,k)$\\ \hline
50 [Hz] & 0.005 [p.u.] & (1, 0.1, 0.01)& (0.01, 10, 7.24)\\\hline
\end{tabular}
\\\vspace*{1mm}
\begin{tabular}{|c|c|c|c|c|c|}
\hline
\multicolumn{6}{|c|}{IBR Capacity + Load Apparent Power and Power Factor} \\ \hline
IBR/Bus \#& 1 & 2 & 3 & 4 & 5 \\ \hline
$S_i^{\rm rated}$ [p.u.] & 1.1 & 0.6 & 0.8 & 0.75 & 1.3 \\ \hline
$S_i^{\rm load}$ [p.u.]  & 0.9 & 0.5 & 0.7 & 0.65 & 1 \\ \hline
${\rm PF}_i$             & 0.85 & 0.9 & 0.88 & 0.92 & 0.87 \\ \hline
\end{tabular}
\\\vspace*{1mm}
\begin{tabular}{|c|c|c|c|c|c|}
\hline
  \multicolumn{3}{|c|}{Bus $i$ to Bus $j$ Interconnection} & \multicolumn{3}{c|}{IBR Output Connection} \\ \hline
  $(i,j)$& $r_{ij}\:[\Omega]$ & $x_{ij}\:[\Omega]$ & IBR \#&  $r_{i}\:[\Omega]$ & $x_{i}\:[\Omega]$ \\ \hline
(1, 2)& 0.2 & 0.3 & 1 & 0.03 & 0.09    \\ \hline
(2, 3)& 0.19 & 0.19 & 2 & 0.1 & 0.25    \\ \hline
(3, 4)& 0.17 & 0.25 & 3 & 0.05 & 0.15   \\ \hline
(4, 5)& 0.15 & 0.22 & 4 & 0.08 & 0.23    \\ \hline
(5, 1)& 0.22 & 0.32 & 5 & 0.07 & 0.2    \\ \hline
\multicolumn{6}{l}{$r$ is resistance and $x$ is reactance.}
\end{tabular}\vspace{-.1in}
\end{table}
\begin{figure*}
\centering
\includegraphics[width=0.495\textwidth]{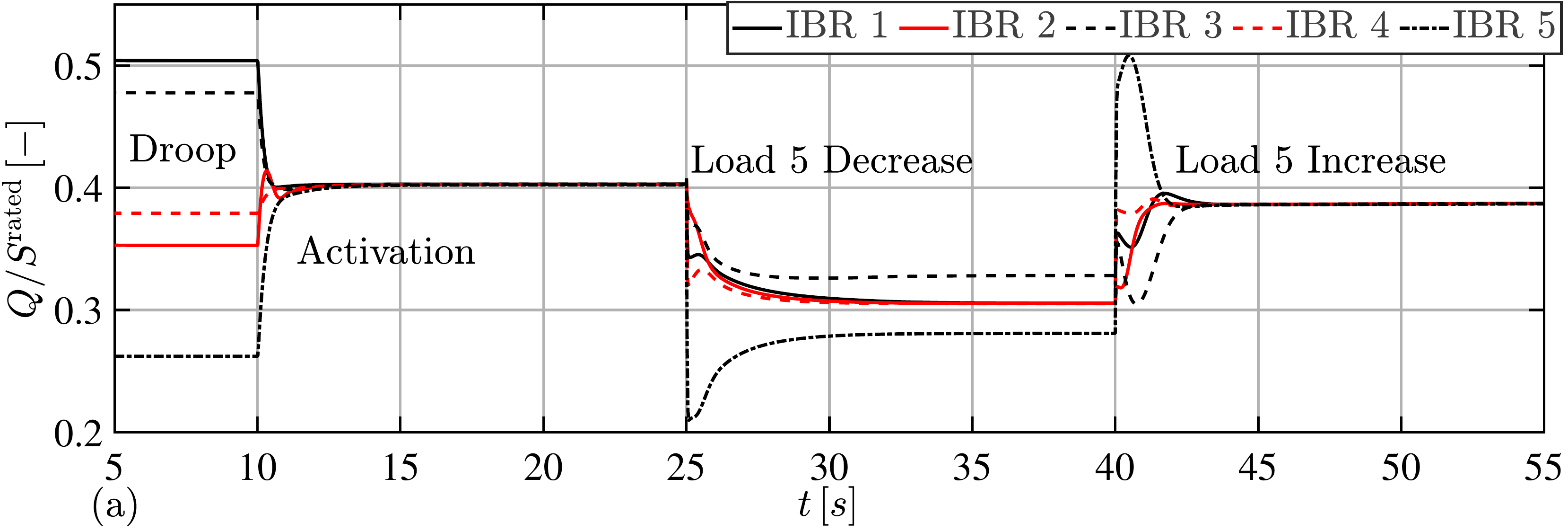}
\includegraphics[width=0.495\textwidth]{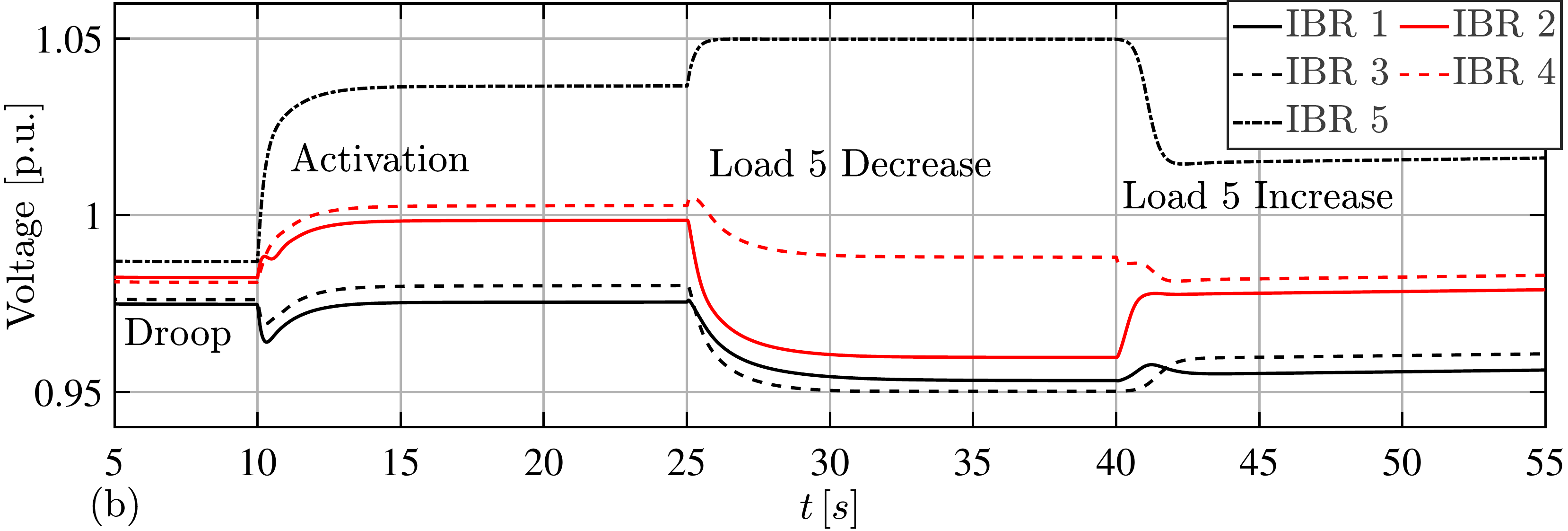}
\includegraphics[width=0.495\textwidth]{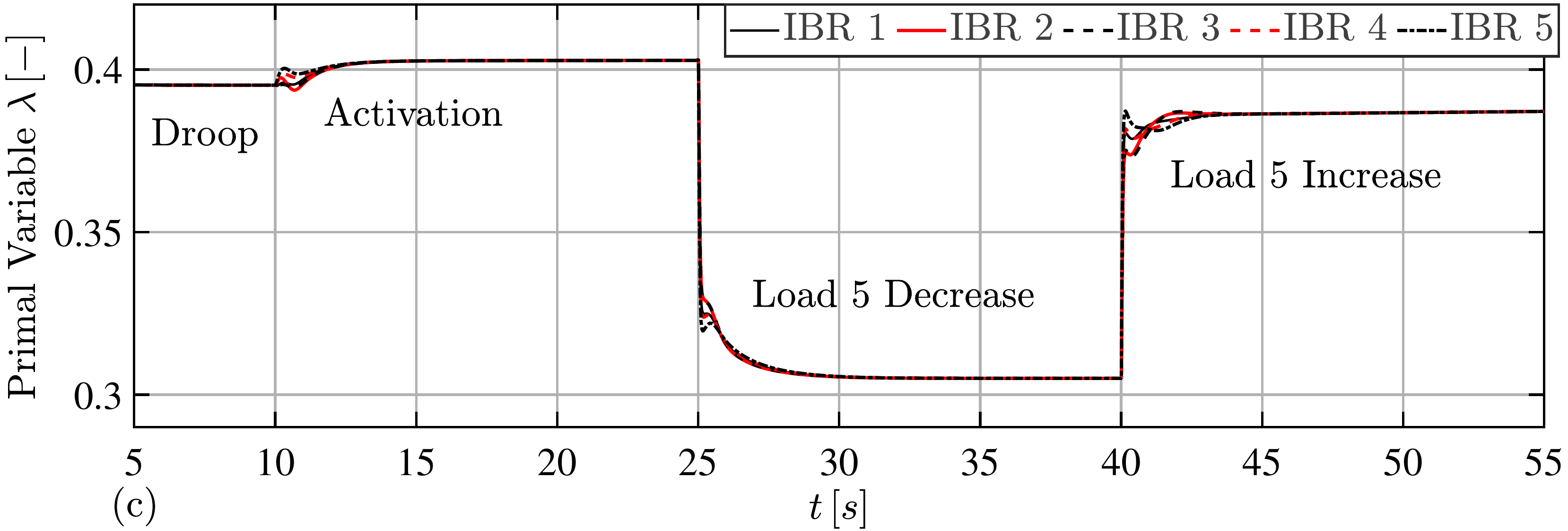}
\includegraphics[width=0.495\textwidth]{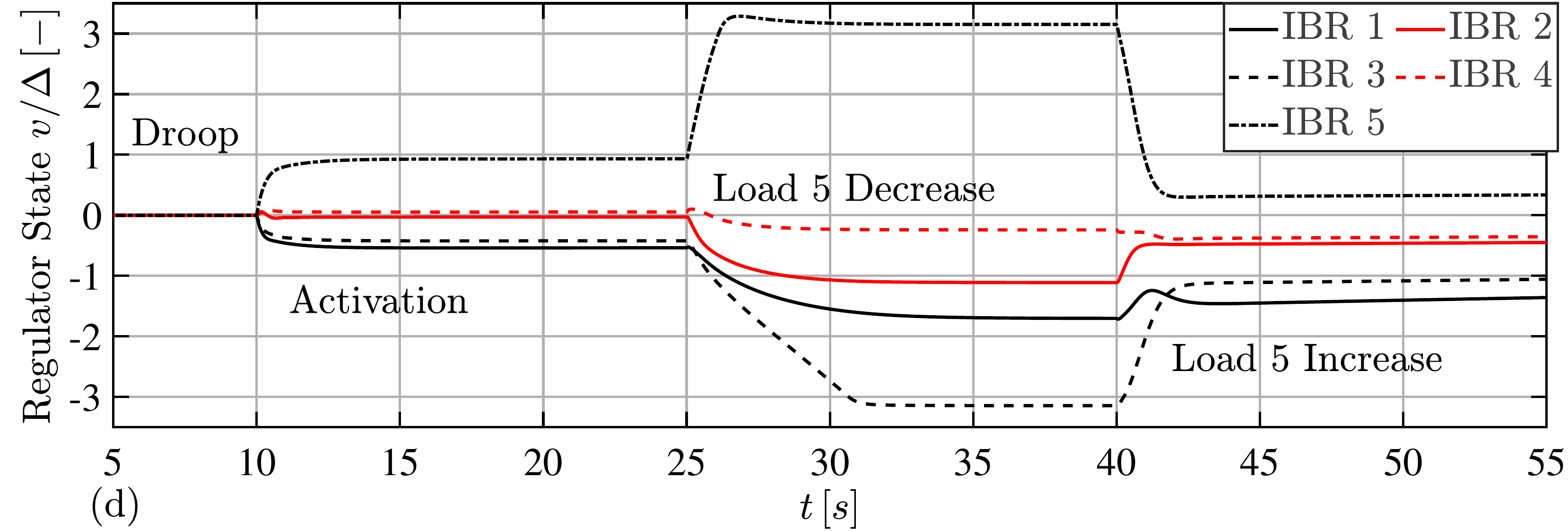}
\includegraphics[width=0.495\textwidth]{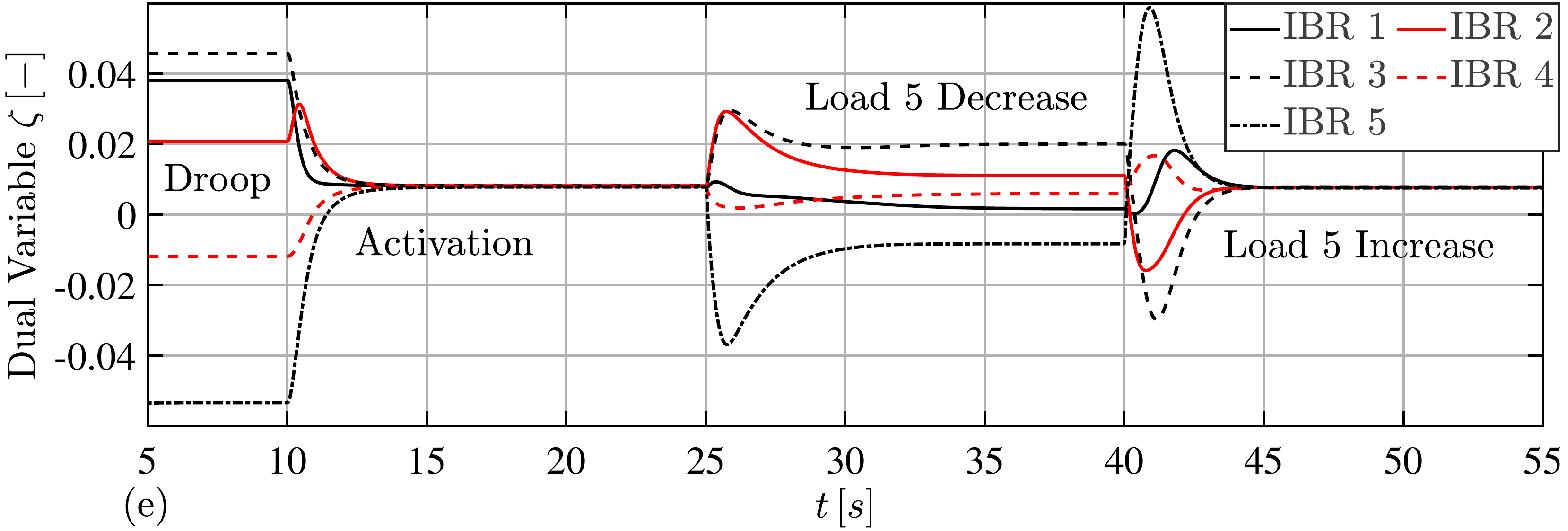}
\includegraphics[width=0.495\textwidth]{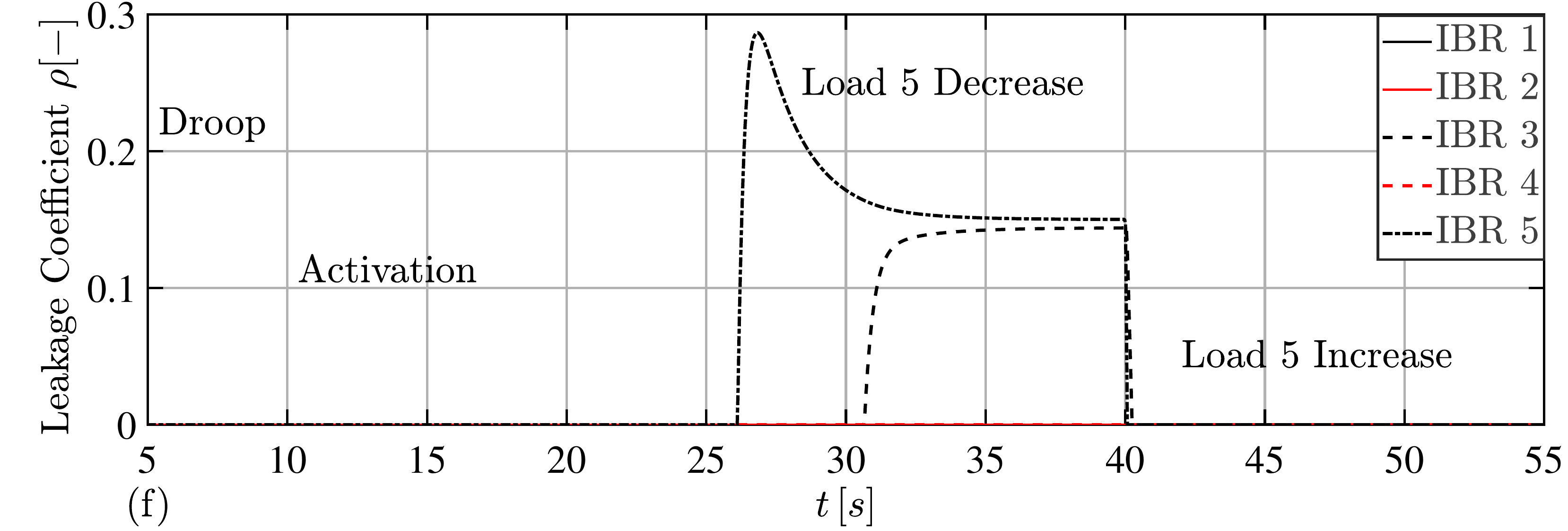}
\includegraphics[width=0.495\textwidth]{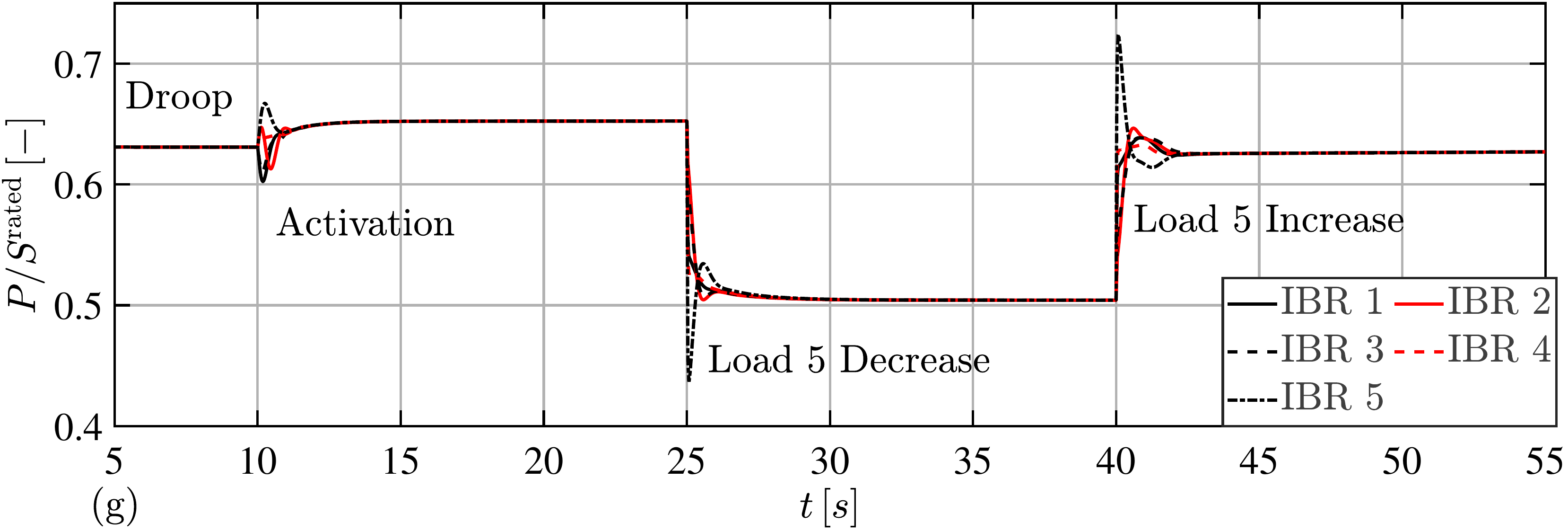}
\includegraphics[width=0.495\textwidth]{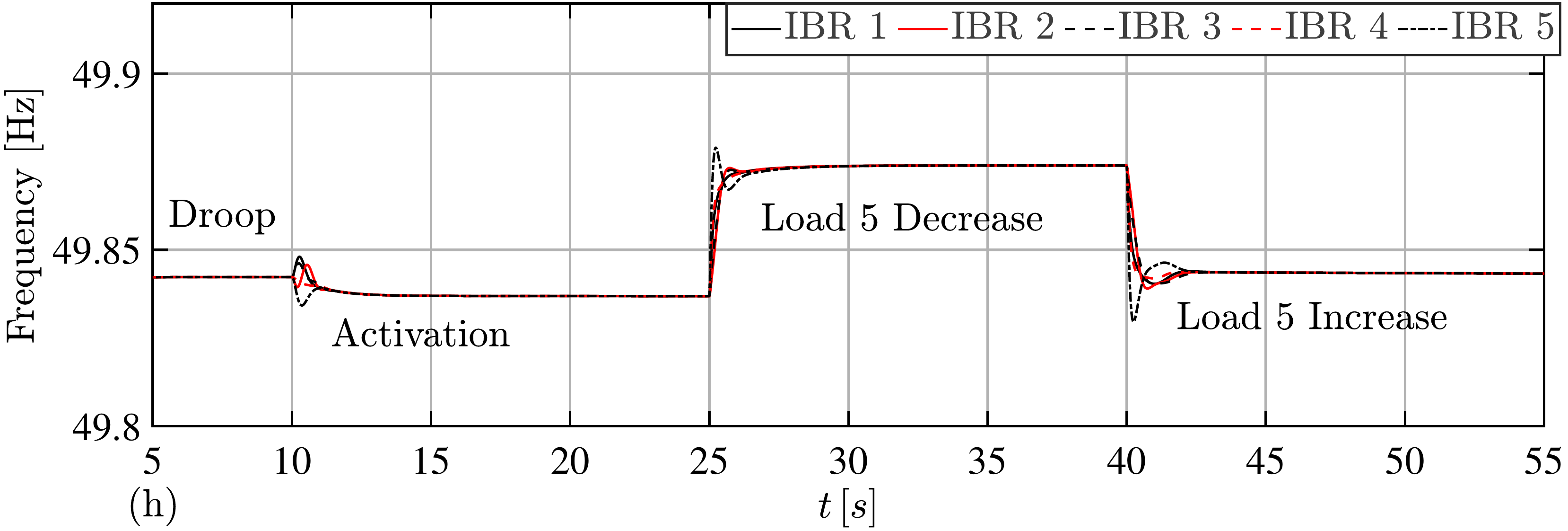}
\caption{Simulation results for Case Study 1; (a) reactive power ratios $Q_i/S_i^{\rm rated}$, (b) voltages $V_i$, (c) primal variables $\lambda_i$, (d) normalized integrator states $v_i/\Delta_i$, (e) dual variables $\zeta_i$, (f) leakage coefficients $\rho_i(v_i)$, (g) active power ratios $P_i/S_i^{\rm rated}$, and frequencies \smash{$f_i=\tfrac{1}{2}\omega_i/\pi$}.\label{Fig:CaseStudy1}\vspace{-.15in}}
\end{figure*}
To verify the effectiveness of the proposed controller, we applied it to a low-voltage 5-bus meshed microgrid system, simulated in MATLAB/Simscape Electrical software environment. The nominal voltage and frequency of the grid are 220-V (RMS) and 50-Hz, respectively. As shown in Fig.~\ref{Fig:LVMG}, the microgrid consists of five local loads energized by five IBRs. Each IBR feeds its corresponding main bus/load via an output connector. Table~\ref{Table:LVMG} shows the electrical and control specifications of the system. It is to note that, in our simulations, we adapted the detailed model of the inverters and internal control loops from\cite{Wu2017VI}.\vspace{-.1in}

\subsection{Case Study 1: Activation and Load Change}
\label{Subsec:CaseStudy1}

We assume that the droop controllers in \eqref{eq:FDroop} and \eqref{eq:VDroop} control the system before activating the proposed controller. According to Fig.~\ref{Fig:CaseStudy1}(a)-(b), we can see that the voltages deviate from the nominal value, and the IBRs do not share the reactive power proportionally. After activation of the controller at $t=10s$, the IBRs start changing their voltages so that their reactive power ratios become equal and, at the same time, their voltages maintain within limits $(0.95,1.05)$ [p.u.]. These results are in line with Properties 2 and 3 in Proposition~\ref{Proposition:SSProperties}.

At $t=25s$, the load at bus number 5 decreases by 80\%, which means to keep the proportional sharing, the 5th IBR must feed the other loads instead of the lost 80\% local load. Therefore, in a collaborative effort to reach an agreement on a new equal power ratio, this IBR increases its voltage, and the other ones reduce their voltages until the IBRs 1, 2, and 4 reach an agreement on $Q_i/S_i^{\rm rated}$. However, the 3rd and 5th IBRs fail to join this agreement because their voltages are already saturated at the minimum and maximum limits, respectively. However, they have come close to the point of agreement and stayed there. We can observe their effort in reaching a consensus with the other IBRs in Fig.~\ref{Fig:CaseStudy1}(d) and Fig. ~\ref{Fig:CaseStudy1}(f). After $t=25s$, the 5th IBR keeps integrating and increasing $v_5/\Delta_5$ to increase its voltage to the maximum. However, as the voltage is saturated using the $\tanh$ function, the voltage does not change much. Therefore, at $t\approx 26s$, when $v_5>3\Delta_5$, the leakage coefficient $\rho_5$ takes a positive value to prevent the integrator wind-up. Meanwhile, the 3rd IBR also keeps integrating but decreasing $h_3/\Delta_3$, until its voltage gets saturated at $t\approx 34s$ and $\rho_3$ also takes a positive value. These results are in line with properties 2 and 4 in Proposition~\ref{Proposition:SSProperties}. After restoring the lost load at $t=40s$, the IBRs re-achieve proportional reactive power sharing, and their leakage coefficients are all restored to zero.
Fig.~\ref{Fig:CaseStudy1}(c) and Fig.~\ref{Fig:CaseStudy1}(e) show the primal-dual variables; we can see that thanks to the dual variables, the primal variable always converges to the average of the reactive power ratios (see \eqref{eq:SSLambda}), no matter if the voltages are saturated or not. Therefore, they are treated as a globally-common variable like frequency (see Fig.~6(h)) and used as a reliable reference for the reactive power ratios of the IBRs. We can also see the active and frequency responses in Fig~\ref{Fig:CaseStudy1}(g)-(h), reflecting the impacts of the voltage-reactive power controller on the frequency-active power dynamics.\vspace{-.1in} 
\begin{figure}
\centering
\includegraphics[width=.494\columnwidth]{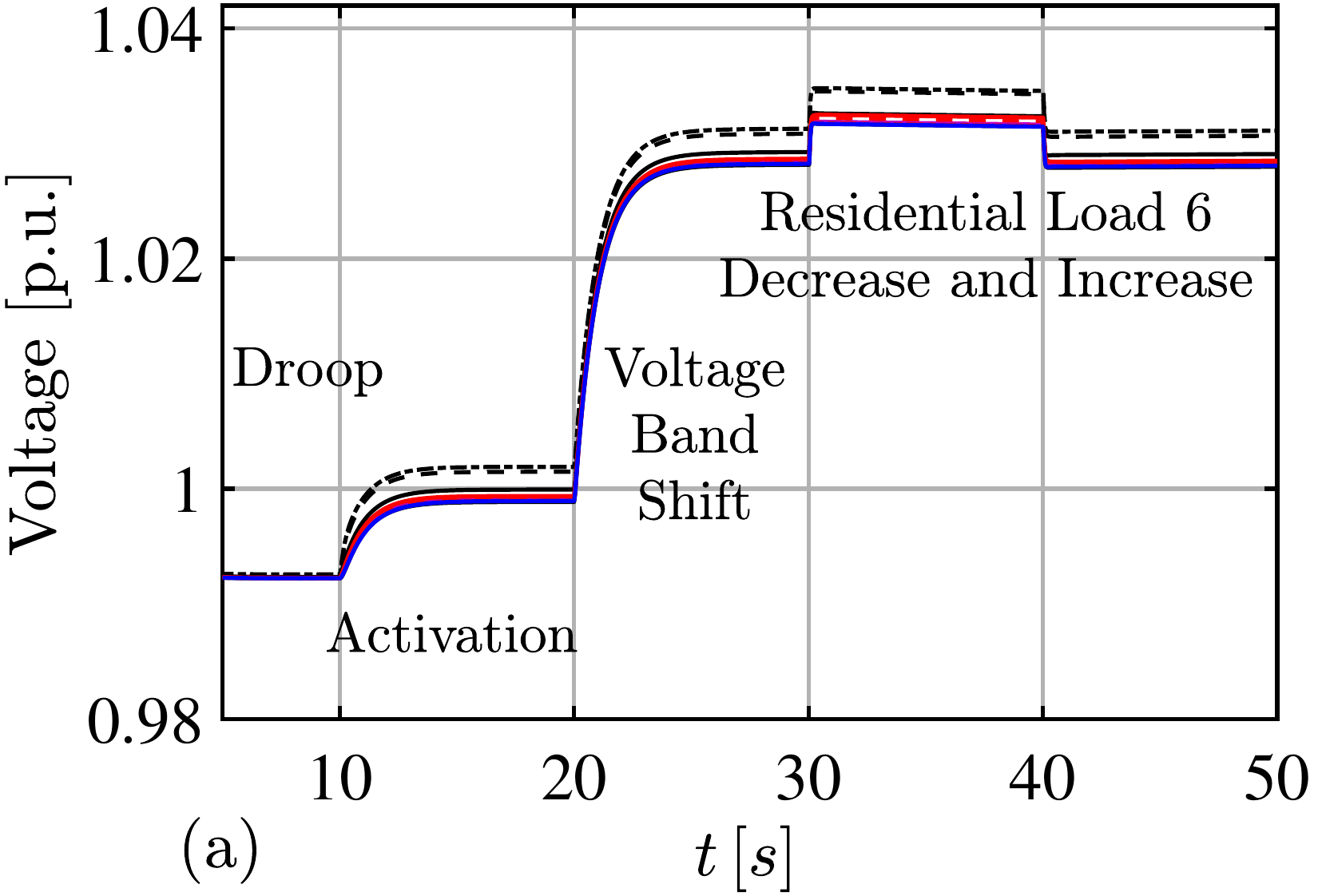}
\includegraphics[width=.494\columnwidth]{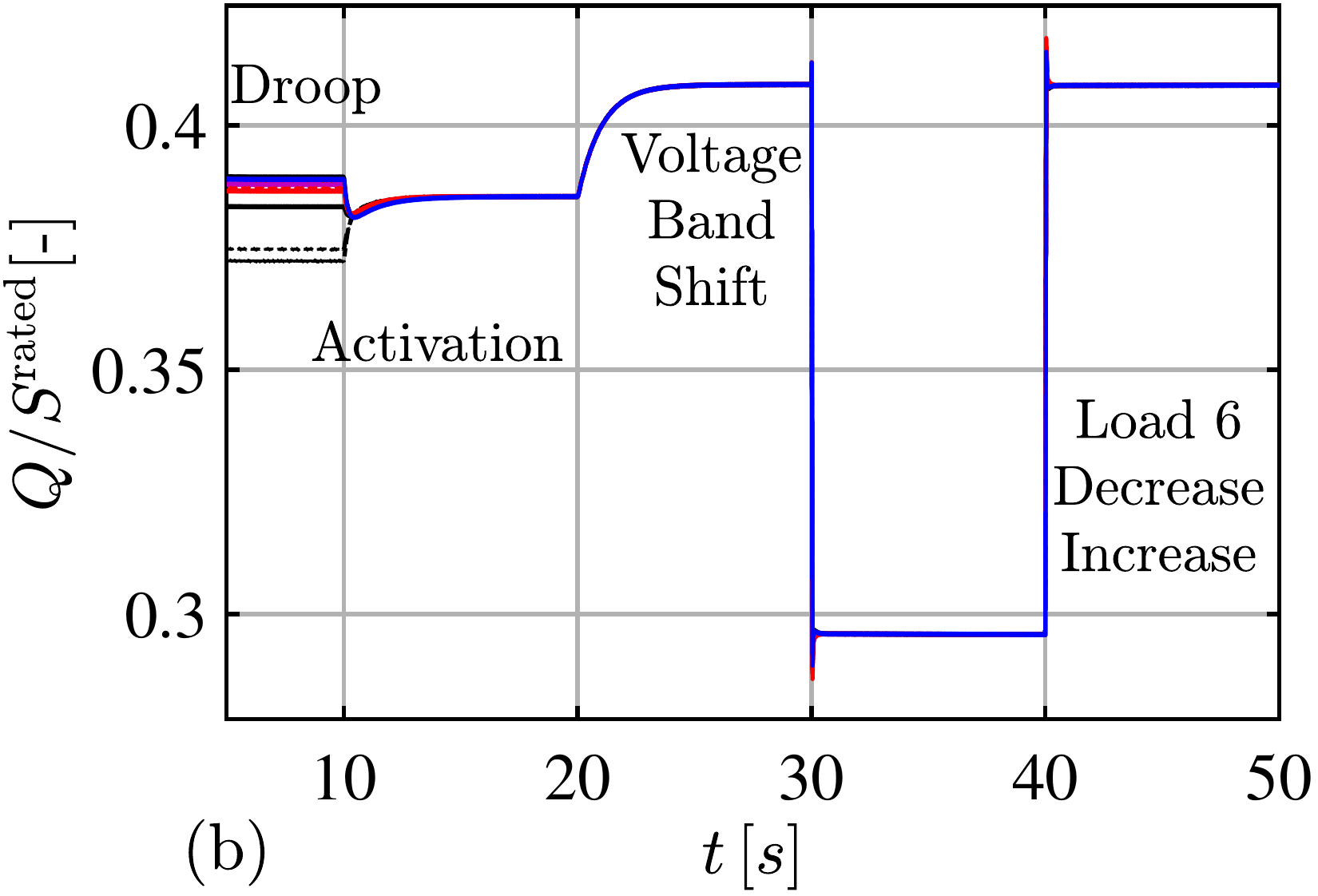}
\caption{Simulation results for Case Study 2; (a) voltages and (b) reactive power ratios $Q_i/S_i^{\rm rated}$.\label{Fig:CaseStudy2}\vspace{-.15in}}
\end{figure}

\subsection{Case Study 2: CIGRE Benchmark and Voltage Level Shift}
\label{Subsec:CaseStudy2}

We also applied our controller to a system based on the Subnetwork 1 of the European medium-voltage (20-kV,50-Hz) distribution network benchmark, provided by CIGRE Task Force C6.04.02\cite{CIGRE575}. Except for the following modifications, all the specifications of the test system are the same as the original network. Based on the Task Force recommendation, the simulated microgrid is an isolated 9-bus subnetwork of the CIGRE system composed of buses number 3 to 11. All the distributed generators at each bus are lumped into one single dispatchable IBR governed by the proposed controller. To meet the maximum load demand in the islanded microgrid the rated power of the IBRs are all increased by 50\%. The control gains are similar to the previous case study, but the voltage limits are $(0.98,1.02)$ [p.u.]. The simulation results are shown in Fig.~\ref{Fig:CaseStudy2}. At $t=10s$, the controller is activated and the voltages and reactive powers are controlled properly. At $t=20s$, we shift the voltage level by setting the new limits $(1.01,1.05)$ [p.u.]. At $t=30s$ and $t=40s$, we disconnect and connect back the residential loads at buses 6 and 8. The results highlight that under the proposed method, we can shift the voltage level in a controlled way while keeping reactive power sharing at different voltage levels. \vspace{-.1in}

\section{Conclusion}
\label{Sec:Conclusion}

Voltage regulation and reactive power sharing in power systems are two highly coupled control objectives. This coupling is because reactive power flow between two nodes depends more strongly on their voltage differences than the absolute values of the voltages. We proposed a nonlinear controller based on a hyperbolic tangent function and a distributed primal-dual optimizer. The controller provides the IBRs with acceptable reactive power sharing while keeping their voltages within some user-defined limits. We also found stability conditions for the system considering the voltage-angle couplings, under timescale separation between the voltage and optimizer dynamics. The numerical simulations, followed by a proposed parameter selection guideline, indicated a promising performance from the proposed method in controlling the voltage level of the network and achieving reactive power sharing among the IBRs.\vspace{-.1in}
\appendices

\section{Communication Network Model and Graph Theory}
\label{Appendix:Graph}
An inter-IBR data network can be modeled by an undirected graph where the IBRs and communication links are considered its nodes and edges, respectively. Let $\mathcal{G}=(\mathcal{N},\mathcal{E},\mathcal{A})$ be a graph with $\mathcal{N}=\{ 1,...,n\}$, $\mathcal{E}\subseteq\mathcal{N} \times \mathcal{N}$, and $\mathcal{A}=[a_{ij}]\in\mathbb{R}^{n\times n}$ being its node set, edge set, and adjacency matrix, respectively. If the nodes $i$ and $j$ directly exchange data, they are neighbors, meaning that $(i,j) \in \mathcal{E}$ and $(j,i)\in\mathcal{E}$, and $a_{ij}=a_{ji}=0$; otherwise,  $a_{ij}=a_{ji}=0$. Let $N_i=\{ j\mid (j,i)\in\mathcal{E}\}$ and $d_i={\sum}_{j\in N_i}a_{ij}$ be the neighbor set and in-degree associated with node $i$, respectively. Laplacian matrix of $\mathcal{G}$ is defined as $\mathcal{L}=\mathcal{D}-\mathcal{A}$, where $\mathcal{D}=\mathrm{diag}\{d_i\}$. A walk (or path) from node $i$ to node $j$ is an ordered sequence of nodes such that any pair of consecutive nodes in the sequence is an edge of the graph. A graph is connected if there exists a walk between any two nodes\cite{FB-LNS}.\vspace{-.1in}
\section{Components of the Reduced Dynamics in \eqref{eq:GroundedSystem}}
\label{Appendix:Components}
With $I_{r}=[0_{n-1}\, I_{n-1}]\in \mathbb{R}^{(n-1)\times n}$, one can obtain the components of the system \eqref{eq:GroundedSystem} as
\[
\begin{cases}
R_\theta = - I_r Tm S^{-1}  J_\theta^P T^{-1} I_r^\top,\qquad R_{\theta V} = -  I_r Tm S^{-1} J_V^P,\\
R_{v\theta}= [V_\star](\mathcal{K}-I_n)S^{-1} J_\theta^Q T^{-1}  I_r^\top,\\
R_{vV} = [V_\star](\mathcal{K}-I_n)S^{-1} J_V^Q,\qquad R_{v\zeta}= - [V_\star]\mathcal{K} \mathcal{L} T^{-1} I_r^\top,\\
R_{\zeta \theta}= I_r \tau_v^{-1}T\mathcal{L} \mathcal{K}S^{-1}   J_\theta^Q T^{-1} I_r^\top,\\
R_{\zeta V} = I_r \tau_v^{-1}T\mathcal{L} \mathcal{K}S^{-1} J_V^Q,\quad R_\zeta = -  I_r \tau_v^{-1}T\mathcal{L} \mathcal{K} \mathcal{L} T^{-1} I_r^\top,
\end{cases}
\]

\[
\begin{cases}
R_\theta^{\rm av} = -1_n^\top T^\top Tm S^{-1}  J_\theta^P T^{-1} I_r^\top,\\
R_{\theta V}^{\rm av} = -  1_n^\top T^\top Tm S^{-1} J_V^P,\\
d_\theta = I_r \omega_{\rm nom}T1_n  -  I_r Tm S^{-1} w_P,\\
d_v=\beta V_\star + [V_\star](\mathcal{K}-I_n)S^{-1} w_Q,\\
d_\zeta = I_r \tau_v^{-1}T\mathcal{L} \mathcal{K}S^{-1} w_Q, \\
d_\theta^{\rm av} = 1_n^\top T^\top \omega_{\rm nom}T1_n  -  1_n^\top T^\top Tm S^{-1} w_P.  
\end{cases}
\]


\bibliographystyle{IEEEtran.bst}
\bibliography{IEEEabrv,References}

\begin{thebibliography}{10}
\providecommand{\url}[1]{#1}
\csname url@samestyle\endcsname
\providecommand{\newblock}{\relax}
\providecommand{\bibinfo}[2]{#2}
\providecommand{\BIBentrySTDinterwordspacing}{\spaceskip=0pt\relax}
\providecommand{\BIBentryALTinterwordstretchfactor}{4}
\providecommand{\BIBentryALTinterwordspacing}{\spaceskip=\fontdimen2\font plus
\BIBentryALTinterwordstretchfactor\fontdimen3\font minus
  \fontdimen4\font\relax}
\providecommand{\BIBforeignlanguage}[2]{{%
\expandafter\ifx\csname l@#1\endcsname\relax
\typeout{** WARNING: IEEEtran.bst: No hyphenation pattern has been}%
\typeout{** loaded for the language `#1'. Using the pattern for}%
\typeout{** the default language instead.}%
\else
\language=\csname l@#1\endcsname
\fi
#2}}
\providecommand{\BIBdecl}{\relax}
\BIBdecl

\bibitem{Yunjie2022}
Y.~Gu and T.~C. Green, ``Power system stability with a high penetration of
  inverter-based resources,'' \emph{Proceedings of the IEEE}, to be published.

\bibitem{IEEEMGTF}
M.~Farrokhabadi \emph{et~al.}, ``Microgrid stability definitions, analysis, and
  examples,'' \emph{{IEEE} Trans. Power Syst.}, vol.~35, no.~1, pp. 13--29,
  Jan. 2020.

\bibitem{KhayatReview}
Y.~Khayat \emph{et~al.}, ``On the secondary control architectures of ac
  microgrids: An overview,'' \emph{{IEEE} Trans. Power Electron.}, vol.~35,
  no.~6, pp. 6482--6500, Jun. 2020.

\bibitem{Babak2019}
B.~Abdolmaleki \emph{et~al.}, ``An instantaneous event-triggered hz–watt
  control for microgrids,'' \emph{{IEEE} Trans. Power Syst.}, vol.~34, no.~5,
  pp. 3616--3625, Sept. 2019.

\bibitem{Simpson2017}
J.~W. Simpson-Porco, F.~Dörfler, and F.~Bullo, ``Voltage stabilization in
  microgrids via quadratic droop control,'' \emph{{IEEE} Trans. Autom.
  Control}, vol.~62, no.~3, pp. 1239--1253, Mar. 2017.

\bibitem{Molzahn2017}
D.~K. Molzahn \emph{et~al.}, ``A survey of distributed optimization and control
  algorithms for electric power systems,'' \emph{{IEEE} Trans. Smart Grid},
  vol.~8, no.~6, pp. 2941--2962, Nov. 2017.

\bibitem{Bidram2013}
A.~Bidram \emph{et~al.}, ``Distributed cooperative secondary control of
  microgrids using feedback linearization,'' \emph{{IEEE} Trans. Power Syst.},
  vol.~28, no.~3, pp. 3462--3470, Aug. 2013.

\bibitem{Shahab2020}
M.~A. Shahab \emph{et~al.}, ``Distributed consensus-based fault tolerant
  control of islanded microgrids,'' \emph{{IEEE} Trans. Smart Grid}, vol.~11,
  no.~1, pp. 37--47, Jan. 2020.

\bibitem{Afshari2022}
A.~Afshari \emph{et~al.}, ``Robust cooperative control of isolated ac
  microgrids subject to unreliable communications: A low-gain feedback
  approach,'' \emph{{IEEE} Syst. J.}, vol.~16, no.~1, pp. 55--66, Mar. 2022.

\bibitem{Babak2020}
B.~Abdolmaleki \emph{et~al.}, ``A zeno-free event-triggered secondary control
  for ac microgrids,'' \emph{{IEEE} Trans. Smart Grid}, vol.~11, no.~3, pp.
  1905--1916, May 2020.

\bibitem{Chen2021}
Y.~Chen \emph{et~al.}, ``Distributed event-triggered secondary control for
  islanded microgrids with proper trigger condition checking period,''
  \emph{{IEEE} Trans. Smart Grid}, 2021.

\bibitem{Guanglei2022}
G.~Zhao, L.~Jin, and Y.~Wang, ``Distributed event-triggered secondary control
  for islanded microgrids with disturbances: A hybrid systems approach,''
  \emph{{IEEE} Trans. Power Syst.}, to be published.

\bibitem{Schiffer2016}
J.~Schiffer \emph{et~al.}, ``Voltage stability and reactive power sharing in
  inverter-based microgrids with consensus-based distributed voltage control,''
  \emph{{IEEE} Trans. Control Syst. Technol.}, vol.~24, no.~1, pp. 96--109,
  Jan. 2016.

\bibitem{Fan2017}
Y.~Fan, G.~Hu, and M.~Egerstedt, ``Distributed reactive power sharing control
  for microgrids with event-triggered communication,'' \emph{{IEEE} Trans.
  Control Syst. Technol.}, vol.~25, no.~1, pp. 118--128, Jan. 2017.

\bibitem{Weng2019}
S.~Weng \emph{et~al.}, ``Distributed event-triggered cooperative control for
  frequency and voltage stability and power sharing in isolated inverter-based
  microgrid,'' \emph{{IEEE} Trans. Cybern.}, vol.~49, no.~4, pp. 1427--1439,
  Apr. 2019.

\bibitem{Li2021}
X.~Li \emph{et~al.}, ``Resilience for communication faults in reactive power
  sharing of microgrids,'' \emph{{IEEE} Trans. Smart Grid}, vol.~12, no.~4, pp.
  2788--2799, Jul. 2021.

\bibitem{Wong2020}
Y.~C.~C. Wong \emph{et~al.}, ``Consensus virtual output impedance control based
  on the novel droop equivalent impedance concept for a multi-bus radial
  microgrid,'' \emph{{IEEE} Trans. Energy Convers.}, vol.~35, no.~2, pp.
  1078--1087, Jun. 2020.

\bibitem{Wong2021}
Y.~C.~C. Wong \emph{et~al.}, ``A consensus-based adaptive virtual output
  impedance control scheme for reactive power sharing in radial microgrids,''
  \emph{{IEEE} Trans. Ind. Appl.}, vol.~57, no.~1, pp. 784--794, Jan./Feb.
  2021.

\bibitem{Zhou2020}
J.~Zhou, M.-J. Tsai, and P.-T. Cheng, ``Consensus-based cooperative droop
  control for accurate reactive power sharing in islanded ac microgrid,''
  \emph{{IEEE} J. Emerg. Sel. Top. Power Electron.}, vol.~8, no.~2, pp.
  1108--1116, Jun. 2020.

\bibitem{Bidram2014}
A.~Bidram, A.~Davoudi, and F.~L. Lewis, ``A multiobjective distributed control
  framework for islanded ac microgrids,'' \emph{{IEEE} Trans. Ind. Informat.},
  vol.~10, no.~3, pp. 1785--1798, Aug. 2014.

\bibitem{Habibi2023}
S.~I. Habibi \emph{et~al.}, ``Multiagent-based nonlinear generalized minimum
  variance control for islanded ac microgrids,'' \emph{{IEEE} Trans. Power
  Syst.}, to be published.

\bibitem{Choi2021}
J.~Choi, S.~I. Habibi, and A.~Bidram, ``Distributed finite-time event-triggered
  frequency and voltage control of ac microgrids,'' \emph{{IEEE} Trans. Power
  Syst.}, vol.~37, no.~3, pp. 1979--1994, May 2022.

\bibitem{Raeispour2021}
M.~Raeispour \emph{et~al.}, ``Robust distributed disturbance-resilient
  $h_\infty$-based control of off-grid microgrids with uncertain
  communications,'' \emph{{IEEE} Syst. J.}, vol.~15, no.~2, pp. 2895--2905,
  Jun. 2021.

\bibitem{Ge2021}
P.~Ge \emph{et~al.}, ``Resilient secondary voltage control of islanded
  microgrids: An eskbf-based distributed fast terminal sliding mode control
  approach,'' \emph{{IEEE} Trans. Power Syst.}, vol.~36, no.~2, pp. 1059--1070,
  Mar. 2021.

\bibitem{Porco2015}
J.~W. Simpson-Porco \emph{et~al.}, ``Secondary frequency and voltage control of
  islanded microgrids via distributed averaging,'' \emph{{IEEE} Trans. Ind.
  Electron.}, vol.~62, no.~11, pp. 7025--7038, Nov. 2015.

\bibitem{Lai2016}
J.~Lai \emph{et~al.}, ``Droop-based distributed cooperative control for
  microgrids with time-varying delays,'' \emph{{IEEE} Trans. Smart Grid},
  vol.~7, no.~4, pp. 1775--1789, Jul. 2016.

\bibitem{Lu2018}
X.~Lu \emph{et~al.}, ``A novel distributed secondary coordination control
  approach for islanded microgrids,'' \emph{{IEEE} Trans. Smart Grid}, vol.~9,
  no.~4, pp. 2726--2740, Jul. 2018.

\bibitem{Wang2019}
Y.~Wang \emph{et~al.}, ``Cyber-physical design and implementation of
  distributed event-triggered secondary control in islanded microgrids,''
  \emph{{IEEE} Trans. Ind. Appl.}, vol.~55, no.~6, pp. 5631--5642, Nov./Dec.
  2019.

\bibitem{Nasirian2016}
V.~Nasirian \emph{et~al.}, ``Droop-free distributed control for ac
  microgrids,'' \emph{{IEEE} Trans. Power Electron.}, vol.~31, no.~2, pp.
  1600--1617, Feb. 2016.

\bibitem{Shafiee2018}
Q.~Shafiee \emph{et~al.}, ``A multi-functional fully distributed control
  framework for ac microgrids,'' \emph{{IEEE} Trans. Smart Grid}, vol.~9,
  no.~4, pp. 3247--3258, Jul. 2018.

\bibitem{Zhou2018}
J.~Zhou \emph{et~al.}, ``Consensus-based distributed control for accurate
  reactive, harmonic, and imbalance power sharing in microgrids,'' \emph{{IEEE}
  Trans. Smart Grid}, vol.~9, no.~4, pp. 2453--2467, Jul. 2018.

\bibitem{Shi2020}
M.~Shi \emph{et~al.}, ``Pi-consensus based distributed control of ac
  microgrids,'' \emph{{IEEE} Trans. Power Syst.}, vol.~35, no.~3, pp.
  2268--2278, May 2020.

\bibitem{Shi2021}
M.~Shi \emph{et~al.}, ``Observer-based resilient integrated distributed control
  against cyberattacks on sensors and actuators in islanded ac microgrids,''
  \emph{{IEEE} Trans. Smart Grid}, vol.~12, no.~3, pp. 1953--1963, May 2021.

\bibitem{Mohiuddin2022}
S.~M. Mohiuddin and J.~Qi, ``Optimal distributed control of ac microgrids with
  coordinated voltage regulation and reactive power sharing,'' \emph{{IEEE}
  Trans. Smart Grid}, vol.~13, no.~3, pp. 1789--1800, May 2022.

\bibitem{Mohiuddin2020}
S.~M. Mohiuddin and J.~Qi, ``Droop-free distributed control for ac microgrids
  with precisely regulated voltage variance and admissible voltage profile
  guarantees,'' \emph{{IEEE} Trans. Smart Grid}, vol.~11, no.~3, pp.
  1956--1967, May 2020.

\bibitem{IEEEStd1547}
``{IEEE} standard for interconnection and interoperability of distributed
  energy resources with associated electric power systems interfaces,''
  \emph{IEEE Std 1547-2018 (Revision of IEEE Std 1547-2003)}, pp. 1--138, 2018,
  doi: 10.1109/IEEESTD.2018.8332112.

\bibitem{Ortmann2020}
L.~Ortmann \emph{et~al.}, ``Fully distributed peer-to-peer optimal voltage
  control with minimal model requirements,'' \emph{Electric Power Systems
  Research}, vol. 189, p. 106717, 2020.

\bibitem{Renke2017}
R.~Han \emph{et~al.}, ``Containment and consensus-based distributed
  coordination control to achieve bounded voltage and precise reactive power
  sharing in islanded ac microgrids,'' \emph{{IEEE} Trans. Ind. Appl.},
  vol.~53, no.~6, pp. 5187--5199, Nov./Dec. 2017.

\bibitem{BabakSEST}
B.~Abdolmaleki and G.~Bergna-Diaz, ``Voltage containment and reactive
  power-sharing in microgrids: Centralized and distributed approaches,'' in
  \emph{2022 International Conference on Smart Energy Systems and Technologies
  (SEST)}, Eindhoven, Netherlands, 2022, pp. 1--6.

\bibitem{Wu2017VI}
X.~Wu, C.~Shen, and R.~Iravani, ``Feasible range and optimal value of the
  virtual impedance for droop-based control of microgrids,'' \emph{{IEEE}
  Trans. Smart Grid}, vol.~8, no.~3, pp. 1242--1251, May 2017.

\bibitem{DArco2014}
S.~D'Arco and J.~A. Suul, ``Equivalence of virtual synchronous machines and
  frequency-droops for converter-based microgrids,'' \emph{{IEEE} Trans. Smart
  Grid}, vol.~5, no.~1, pp. 394--395, Jan. 2014.

\bibitem{WeiDu2021}
W.~Du \emph{et~al.}, ``Modeling of grid-forming and grid-following inverters
  for dynamic simulation of large-scale distribution systems,'' \emph{{IEEE}
  Trans. Power Del.}, vol.~36, no.~4, pp. 2035--2045, Aug. 2021.

\bibitem{Mobin2021}
M.~Naderi \emph{et~al.}, ``Low-frequency small-signal modeling of
  interconnected ac microgrids,'' \emph{{IEEE} Trans. Power Syst.}, vol.~36,
  no.~4, pp. 2786--2797, Jul. 2021.

\bibitem{Babak2020TIE}
B.~Abdolmaleki and Q.~Shafiee, ``Online kron reduction for economical frequency
  control of microgrids,'' \emph{{IEEE} Trans. Ind. Electron.}, vol.~67,
  no.~10, pp. 8461--8471, Oct. 2020.

\bibitem{Kundur1994}
P.~Kundur, J.~B. Neal, and G.~L. Mark, \emph{Power System Stability and
  Control}.\hskip 1em plus 0.5em minus 0.4em\relax New York, NY, USA:
  McGraw-Hill, 1994.

\bibitem{Boyd}
S.~Boyd and L.~Vandenberghe, \emph{Convex Optimization}.\hskip 1em plus 0.5em
  minus 0.4em\relax USA: Cambridge University Press, 2004.

\bibitem{Cherukuri2016}
A.~Cherukuri, E.~Mallada, and J.~Cortés, ``Asymptotic convergence of
  constrained primal–dual dynamics,'' \emph{Systems \& Control Letters},
  vol.~87, pp. 10--15, Jan. 2016.

\bibitem{FB-LNS}
\BIBentryALTinterwordspacing
F.~Bullo, \emph{Lectures on Network Systems}, 1st~ed.\hskip 1em plus 0.5em
  minus 0.4em\relax Kindle Direct Publishing, 2022. [Online]. Available:
  \url{http://motion.me.ucsb.edu/book-lns}
\BIBentrySTDinterwordspacing

\bibitem{Khalil}
H.~K. Khalil, \emph{Nonlinear Systems}, 3rd~ed.\hskip 1em plus 0.5em minus
  0.4em\relax Englewood Cliffs, NJ, USA: Prentice Hall, 2002.

\bibitem{DorflerOverDamped}
F.~D\"{o}rfler and F.~Bullo, ``Synchronization and transient stability in power
  networks and nonuniform kuramoto oscillators,'' \emph{SIAM Journal on Control
  and Optimization}, vol.~50, no.~3, pp. 1616--1642, 2012.

\bibitem{SimpsonPorco2016}
J.~W. Simpson-Porco and F.~Bullo, ``Distributed monitoring of voltage collapse
  sensitivity indices,'' \emph{{IEEE} Trans. Smart Grid}, vol.~7, no.~4, pp.
  1979--1988, Jul. 2016.

\bibitem{CIGRE575}
K.~Strunz \emph{et~al.}, ``Benchmark systems for network integration of
  renewable and distributed energy resources,'' Task Force C6.04, CIGRE, Paris,
  France, Technical Brochure 575, 2014.

\end{thebibliography}

\end{document}